\documentclass[11pt, letterpaper]{article}
\usepackage{fullpage}
\usepackage{fullpage}

\usepackage[T1]{fontenc}

\usepackage[colorlinks=true,pdfpagemode=none,linkcolor=blue,citecolor=blue]{hyperref}
\usepackage{amsmath,amsfonts,amsthm}
\usepackage[figure,boxed,lined]{algorithm2e}
\usepackage{color}
\usepackage{multirow}
\usepackage{enumerate}
\usepackage{graphicx}

\newtheorem{theorem}{Theorem}[section]

\newtheorem{lemma}[theorem]{Lemma}

\newtheorem{fact}[theorem]{Fact}

\newtheorem{invariant}[theorem]{Invariant}

\newcommand{\cenergy}{C_{\mathcal{E}}}
\newcommand{\crho}{C_{\rho^*}}
\newcommand{\cincrease}{C_{B}}
\newcommand{\cdrop}{C_{D}}

\newcommand{\bbR}{\mathbb{R}}

\newcommand{\ceil}[1]{\lceil #1 \rceil}

\newcommand{\norm}[2]{\|#1\|_{#2}}

\newcommand{\inorm}[1]{\|#1\|_{\infty}}

\newcommand{\tO}[1]{\widetilde{O}\left(#1\right)}

\newcommand{\hG}{\widehat{G}}

\newcommand{\hy}{\widehat{y}}

\newcommand{\hu}{\widehat{u}}

\newcommand{\energy}[2]{\mathcal{E}_{#1}(#2)}

\newcommand{\eps}{\varepsilon}

\newcommand{\tf}{\widetilde{f}}
\newcommand{\hf}{\widehat{f}}

\newcommand{\hr}{\widehat{r}}

\newcommand{\tu}{\widetilde{u}}

\newcommand{\hrho}{\widehat{\rho}}
\newcommand{\hgamma}{\widehat{\gamma}}
\newcommand{\tgamma}{\widetilde{\gamma}}

\newcommand{\htheta}{\widehat{\theta}}
\newcommand{\hdelta}{\hat{\delta}}

\newcommand{\vphi}{\boldsymbol{\mathit{\phi}}}

\newcommand{\vrho}{\boldsymbol{\mathit{\rho}}}
\newcommand{\hvrho}{\boldsymbol{\mathit{\widetilde{\rho}}}}

\newcommand{\vgamma}{\boldsymbol{\mathit{\gamma}}}
\newcommand{\vtheta}{\boldsymbol{\mathit{\theta}}}

\newcommand{\hvtheta}{\boldsymbol{\mathit{\widehat{\theta}}}}

\newcommand{\vsigma}{\boldsymbol{\mathit{\sigma}}}
\newcommand{\vvarsigma}{\boldsymbol{\mathit{\varsigma}}}

\newcommand{\hvsigma}{\boldsymbol{\mathit{\widehat{\sigma}}}}
\newcommand{\hvgamma}{\boldsymbol{\mathit{\widehat{\gamma}}}}
\newcommand{\tvgamma}{\boldsymbol{\mathit{\widetilde{\gamma}}}}

\newcommand{\vkappa}{\boldsymbol{\mathit{{\kappa}}}}
\newcommand{\vchi}{\boldsymbol{\mathit{{\chi}}}}
\newcommand{\huu}{\boldsymbol{\mathit{{\widehat{u}}}}}

\newcommand{\tphi}{\widetilde{\phi}}
\newcommand{\hvphi}{\boldsymbol{\widehat{\phi}}}
\newcommand{\tvphi}{\boldsymbol{\tilde{\phi}}}

\newcommand{\bb}{\boldsymbol{\mathit{b}}}

\newcommand{\ff}{\boldsymbol{\mathit{f}}}
\renewcommand{\gg}{\boldsymbol{\mathit{g}}}
\newcommand{\tff}{\boldsymbol{\mathit{\widetilde{f}}}}

\newcommand{\ozz}{\boldsymbol{\mathit{\bar{z}}}}

\newcommand{\hff}{\boldsymbol{\mathit{\widehat{f}}}}
\newcommand{\hyy}{\boldsymbol{\mathit{\widehat{y}}}}

\newcommand{\ogg}{\boldsymbol{\mathit{\bar{g}}}}
\renewcommand{\gg}{\boldsymbol{\mathit{g}}}
\newcommand{\hh}{\boldsymbol{\mathit{h}}}

\newcommand{\rr}{\boldsymbol{\mathit{r}}}
\newcommand{\hrr}{\boldsymbol{\mathit{\widehat{r}}}}

\newcommand{\uu}{\boldsymbol{\mathit{u}}}

\newcommand{\yy}{\boldsymbol{\mathit{y}}}
\newcommand{\zz}{\boldsymbol{\mathit{z}}}

\newcommand{\LL}{\boldsymbol{\mathit{L}}}

\begin{document}

\clubpenalty=10000
\widowpenalty = 10000

\title{Computing Maximum Flow with Augmenting Electrical Flows}

\author{Aleksander M\k{a}dry\thanks{Supported by NSF grant CCF-1553428, and a Sloan Research Fellowship.}\\
       {MIT}\\
       { madry@mit.edu}}
\date{}

\maketitle
\begin{abstract}
We present an $\tO{m^{\frac{10}{7}}U^{\frac{1}{7}}}$-time algorithm for the maximum $s$-$t$ flow problem (and the minimum $s$-$t$ cut problem) in directed graphs with $m$ arcs and largest integer capacity $U$. This matches the running time of the $\tO{(mU)^{\frac{10}{7}}}$-time algorithm of M\k{a}dry \cite{Madry13} in the unit-capacity case, and improves over it, as well as over the $\tO{m
\sqrt{n} \log U}$-time algorithm of  Lee and Sidford \cite{LeeS14}, whenever $U$ is moderately large and the graph is sufficiently sparse. 

By well-known reductions, this also gives us an $\tO{m^{\frac{10}{7}}B^{\frac{1}{7}}}$-time algorithm for the maximum-cardinality  bipartite $\bb$-matching problem in which the largest integer demand is $B$. This, again, matches the $\tO{(mB)^{\frac{10}{7}}}$-time algorithm of M\k{a}dry \cite{Madry13}, when $B=1$, which corresponds to the maximum-cardinality bipartite matching problem, and outperforms it, as well as the $\tO{m
\sqrt{n} \log B}$-time algorithm of Lee and Sidford \cite{LeeS14}, for moderate values of $B$ and sufficiently sparse graphs.

One of the advantages of our algorithm is that it is significantly simpler than the ones presented in \cite{Madry13} and \cite{LeeS14}. In particular, these algorithms employ a sophisticated interior-point method framework, while our algorithm is cast directly in the classic augmenting path setting that almost all the combinatorial maximum flow algorithms use. At a high level, the presented algorithm takes a primal dual approach in which each iteration uses electrical flows computations both to find an augmenting $s$-$t$ flow in the current residual graph and to update the dual solution. We show that by maintain certain careful coupling of these primal and dual solutions we are always guaranteed to make significant progress.

\end{abstract}

\thispagestyle{empty}
\newpage
\setcounter{page}{1}

\section{Introduction}

  The maximum $s$-$t$ flow problem and its dual, the minimum $s$-$t$ cut problem,
  are two of the most fundamental and extensively studied graph problems in combinatorial optimization ~\cite{Schrijver03,AhujaMO93,Schrijver02}. They have a wide range of applications (see~\cite{AhujaMOR95}), are often used as subroutines in other algorithms (see, e.g., \cite{AroraHK05,Sherman09}), and a number of other important problems  -- e.g., bipartite matching problem \cite{CormenLRS09} -- can be reduced to them.  Furthermore, these two problems were often a testbed for development of fundamental algorithmic tools and concepts. Most prominently, the Max-Flow Min-Cut theorem \cite{EliasFS56,FordF56} constitutes the prototypical primal-dual relation.
  
  Several decades of extensive work resulted in a number of developments on these problems (see Goldberg and Rao \cite{GoldbergR98} for an overview) and  many of their generalizations and special cases. Still, despite all this effort, the basic problem of computing maximum $s$-$t$ flow and minimum $s$-$t$ cut in general graphs resisted progress for a long time. In particular, for a number of years, the best running time bound for the problem  was an $O(m\min\{m^{\frac{1}{2}},n^{\frac{2}{3}}\}\log (n^2/m) \log U)$ (with $U$ denoting the largest integer arc capacity)  bound established in a a breakthrough paper by Goldberg and Rao \cite{GoldbergR98} and this bound,  in turn, matched the $O(m\min\{m^{\frac{1}{2}},n^{\frac{2}{3}}\})$ bound for unit-capacity graphs due to Even and Tarjan \cite{EvenT75} -- and, independently, Karzanov \cite{Karzanov73} -- that were put forth more than 40 years ago. 
  
  The above bounds were improved only fairly recently. Specifically, in 2013, M\k{a}dry \cite{Madry13} presented an interior-point method based framework for flow computations that gave an $\tO{m^{\frac{10}{7}}}$-time\footnote{We recall that $\tO{f}$ denotes $O(f \log^c f)$, for some constant $c$.} algorithm for the unit-capacity case of the maximum $s$-$t$ flow and minimum $s$-$t$ cut problems. This finally broke the long-standing $\tO{n^{\frac{3}{2}}}$ running time barrier for sparse graphs, i.e., for $m=O(n)$. Later on, Lee and Sidford \cite{LeeS14} developed a variant of interior-point method that enabled them to obtain improvement for the regime of dense graphs. In particular, their algorithm is able to compute the (general) maximum $s$-$t$ flow and minimum $s$-$t$ cut in $\tO{m\sqrt{n}\log U}$ time and thus improve over the Goldberg-Rao bound whenever the input graph is sufficiently dense.

It is also worth mentioning that, as a precursor to the above developments, substantial progress was made in the context of $(1-\eps)$-approximate variant of the maximum $s$-$t$ flow problem in undirected graphs. In 2011, Christiano et al. \cite{ChristianoKMST11} developed an algorithm that allows one to compute a $(1+\eps)$-approximation to the undirected maximum $s$-$t$ flow (and the minimum $s$-$t$ cut) problem in $\tO{mn^{\frac{1}{3}} \eps^{-11/3}}$ time. Their result relies on devising a new approach to the problem that combines electrical flow computations with multiplicative weights update method (see \cite{AroraHK05}). Later, Lee et al. \cite{LeeRS13} presented a quite different -- but still electrical-flow-based -- algorithm that employs purely gradient-descent-type view to obtain an $\tO{mn^{1/3}\eps^{-2/3}}$-time $(1+\eps)$-approximation for the case of unit capacities. Very recently, this line of work was culminated by Sherman \cite{Sherman13} and Kelner et al. \cite{KelnerLOS14} who independently showed how to integrate non-Euclidean gradient-descent methods with fast poly-logarithmic-approximation algorithms for cut problems of M\k{a}dry \cite{Madry10b} to get an $O(m^{1+o(1)}\eps^{-2})$-time $(1+\eps)$-approximation to the undirected maximum flow problem. Then, Peng \cite{Peng16} built on these works to obtain a truly nearly-linear, i.e.,  $\tO{m\eps^{-2}}$, running time.
  
Finally, we note that, in parallel to the above work that is focused on designing weakly-polynomial algorithms for the maximum $s$-$t$ flow and minimum $s$-$t$ cut problems, there is also a considerable interest in obtaining running time bounds that are strongly-polynomial, i.e., that do not depend on the values of arc capacities. The current best such bound is $O(mn)$ and it follows by combining the algorithms of King et al. \cite{KingRT94} and Orlin \cite{Orlin13}. 
  
  %
  
  
  
  \paragraph{Bipartite Matching Problem.}
  
  Another problem that is related to the maximum $s$-$t$ problem -- and, in fact, can be reduced to it -- is the (maximum-cardinality) bipartite matching problem. This problem is a fundamental assignment task with numerous applications (see, e.g., \cite{AhujaMO93,LovaszP86}) and long history that has its roots in the works of Frobenius \cite{Frobenius12,Frobenius17} and K\"onig \cite{Konig15,Konig16,Konig23} from the early 20th century (see \cite{Schrijver05}). Already in 1931, K\"onig \cite{Konig31} and Egerv\'ary \cite{Egervary31} provided first constructive characterization of maximum matchings in bipartite graphs. This characterization can be turned into a polynomial-time algorithm. Then, in 1973, Hopcroft and Karp \cite{HopcroftK73} and, independently, Karzanov \cite{Karzanov73}, devised the celebrated $O(m\sqrt{n})$-time algorithm. For 40 years this bound remained the best one known in the regime of relatively sparse graphs. Only recently M\k{a}dry \cite{Madry13} obtained an improved, $\tO{m^{10/7}}$ running time. It turns out, however, that whenever the input graph is dense, i.e., when $m$ is close to $n^2$ even better bounds can be obtain. In this setting, one can combine the algebraic approach of Rabin and Vazirani \cite{RabinV89} -- that itself builds on the work of Tutte \cite{Tutte47} and Lov\'asz \cite{Lovasz79} -- with matrix-inversion techniques of Bunch and Hopcroft \cite{BunchH74} to get an algorithm that runs in $O(n^{\omega})$ time (see \cite{Mucha05}), where $\omega\leq 2.3727$ is the exponent of matrix multiplication \cite{CoppersmithW90, Vassilevska12}. Also, later on, Alt et al. \cite{AltBMP91}, and Feder and Motwani \cite{FederM95} developed combinatorial algorithms that offer a slight improvement -- by a factor of, roughly, $\log n/\log \frac{n^2}{m}$ -- over the $O(m\sqrt{n})$ bound of Hopcroft and Karp whenever the graph is sufficiently dense. 
  
  Finally, a lot of developments has been done in the context of the (maximum-cardinality) matching problem in general, i.e., not necessarily bipartite, graphs. Starting with the pioneering work of Edmonds \cite{Edmonds65}, these developments led to bounds that essentially match the running time guarantees that were previously known only for bipartite case. More specifically, the running time bound of $O(m\sqrt{n})$ for the general-graph case was obtained by Micali and Vazirani \cite{MicaliV80,Vazirani94} (see also \cite{GabowT91} and \cite{GoldbergK04}). Then, Mucha and Sankowski \cite{MuchaS04} gave an $O(n^{\omega})$-time algorithms for general graphs that builds on the algebraic characterization of the problem due to Rabin and Vazirani \cite{RabinV89}. This result was later significantly simplified by Harvey \cite{Harvey09}.
  
\subsection{Our Contribution}
In this paper, we put forth a new algorithm for solving the maximum $s$-$t$ flow and the minimum $s$-$t$ cut problems in directed graphs. More precisely, we develop an algorithm that  computes the maximum $s$-$t$ flow of an input graph in time  $\tO{m^{\frac{10}{7}}U^{\frac{1}{7}}}$, where $m$ denotes the number of arcs of that graph and $U$ its largest integer capacity. Known reductions imply similar running time bounds for the minimum $s$-$t$ cut problem as well as for the maximum-cardinality bipartite $\bb$-matching problem, a natural generalization of the maximum bipartite matching problem in which each vertex $v$ has a degree demand $b_v$. For that problem, our algorithm yields an $\tO{m^{\frac{10}{7}}B^{\frac{1}{7}}}$-time algorithm, with $B$ being the largest (integer) vertex demand. 

In the light of the above, for the unit-capacity/demand cases, the resulting algorithms match the performance of the algorithm of M\k{a}dry \cite{Madry13}. The latter algorithm, however, runs in $\tO{{mU}^{\frac{10}{7}}}$ time (which translates into an $\tO{(mB)^{\frac{10}{7}}}$ running time for the bipartite $\bb$-matching problem)  in the case of arbitrary capacities/demands. Consequently, the significantly better dependence of the running time of our algorithm on the largest capacity $U$/ largest demand $B$ makes it much more favorable in that setting. In fact, even though that dependence on $U$/$B$ is still polynomial it enables our algorithm to remain competitive, for a non-trivial range of parameters, with the best existing algorithms that run in time that is logarithmic in $U$/$B$, such as the $\tO{m\sqrt{n}\log U}$-time algorithm of Lee and Sidford \cite{LeeS14}.

Even more crucially, the key advantage of our algorithm is that it is significantly simpler than both the algorithm of M\k{a}dry \cite{Madry13} and that of Lee and Sidford \cite{LeeS14}). Both these algorithms rely heavily on the interior-point method framework. Specifically, \cite{Madry13} designed a certain new variant of path-following interior-point method algorithm for the near-perfect bipartite $\bb$-matching problem that encoded the input maximum $s$-$t$ flow instance. It then used  electrical flow computations to converge to the near-optimal solution for that problem. In order to break the bottlenecking $\tO{m^{\frac{1}{2}}}$ iteration bound, however, M\k{a}dry \cite{Madry13} needed to, first, develop an extensive toolkit for perturbing and preconditioning the underlying electrical flow computation and, then, to combine this machinery with a very careful and delicate analysis of the resulting dynamics. 

Our algorithm also relies on electrical flow computations but it abandons the above methodology and works instead fully within the classic augmenting path framework that almost all the previous combinatorial maximum $s$-$t$ flow algorithms used. In this framework, developed by Ford and Fulkerson \cite{FordF56} (see also \cite{EliasFS56}), the flow is built in stages. Each stage corresponds to finding a so-called augmenting flow in the current residual graph, which is a directed graph that encodes the solution found so far. The algorithm terminates when the residual graph admits no more augmenting flows, i.e., there is no path from $s$ to $t$ in it, since in this case the solution found so far has to be already optimal. 

The chief bottleneck in the running time analysis of augmenting path based algorithms is ensuring that each flow augmentation stage makes sufficient progress. Specifically, one wants to obtain a good trade off between the amount of flow pushed in each augmentation step and the time needed to implement each such flow push. One simple approach is to just use here $s$-$t$ path computations. This is a nearly-linear time procedure but it only guarantees pushing one unit of flow each time. A much more sophisticated primitive developed in this context are blocking flow computations. Combining this primitive with a simple duality argument enabled Goldberg and Rao \cite{GoldbergR98}, who built on the work of Even and Tarjan \cite{EvenT75}, to obtain an $O(m\min\{m^{\frac{1}{2}},n^{\frac{2}{3}}\}\log U)$-time maximum flow that remained the best known algorithm for nearly two decades. Unfortunately, trying to improve such blocking flow-based approaches turned out to be extremely difficult and no progress was made here so far, even in the unit-capacity case for which the best known bounds were established over 40 years ago.

One of the key contributions of this paper is bringing a new type of primitive: electrical flows to the augmenting path framework; and showing how to successfully use it to outperform the blocking flow-based methods. Specifically, our algorithm finds augmenting flows by computing electrical flows in certain symmetrization of the current residual graph -- see Section \ref{sec:simple_analysis} for more details. (Note that performing such a symmetrization is necessary as residual graphs are inherently directed while electrical flows are inherently undirected.) The key difficulty that arises in this context, however, is ensuring that this symmetrized residual graph can still support a significant fraction of the $s$-$t$ capacity of the original residual graph. It is not hard to see that, in general, this might not be the case. To address this problem we introduce a certain careful coupling of the primal and dual solutions, which is inspired by the so-called centrality condition arising in interior-point method based maximum flow algorithms (see \cite{Madry13}). We then show that maintaining this coupling and applying a simple preconditioning technique let us guarantee that looking only for the flows in the symmetrized version of the residual graph still provides sufficient progress in each iteration and, in particular, immediately delivers a $\tO{m^{\frac{3}{2}}\log U}$-time algorithm.

We then (see Section \ref{sec:improved_algorithm}) build on that basic algorithm and develop an $\ell_p$-geometric understanding of its running time analysis. This understanding guides us towards a simple electrical flow perturbation technique -- akin to the perturbation techniques used in \cite{ChristianoKMST11} and \cite{Madry13} -- that enables us to break the $\Omega(\sqrt{m})$ iterations bottleneck that all the blocking flow-based algorithms were suffering from, and thus get the final, improved result. 

We believe that further study of this new augmenting flow based framework will deliver even faster and simpler algorithms.

\subsection{Organization of the Paper}

We begin the technical part of the paper in Section \ref{sec:preliminaries} by presenting some preliminaries on the maximum flow problem and  the notion of electrical flows.  Then, in Section \ref{sec:basic_algorithm}, we present our framework and demonstrate how it yields an $\tO{m^{\frac{3}{2}}\log U}$-time maximum $s$-$t$ flow algorithm. Finally, in Section \ref{sec:improved_algorithm}, we show how to refine our basic framework to obtain the improved running time of $\tO{m^{\frac{10}{7}}U^{\frac{1}{7}}}$.
\section{Preliminaries}\label{sec:preliminaries}




Throughout this paper, we will be viewing graphs as having both lower and upper capacities. Specifically, we will denote by $G=(V,E,\uu)$ a directed graph with a vertex set $V$, an arc set $E$ (we allow parallel arcs), and two (non-negative) integer capacities $u_e^-$ and $u_e^+$, for each arc $e\in E$. (We will explain the role of these capacities below.) Usually, $m$ will denote the number $|E|$ of arcs of the graph in question and $n=|V|$ will be the number of its vertices. We view each arc $e$ of $G$ as an ordered pair $(u,v)$, where $u$ is its {\em tail} and $v$ is its {\em head}.

Observe that this perspective enables us to view undirected graphs as directed ones in which the ordered pair $(u,v)\in E$ is an (undirected) {\em edge} $(u,v)$ and the order just specifies the {\em orientation} of that edge (from $u$ to $v$).


\paragraph{Maximum Flow Problem.} 

The basic notion of this paper is the notion of a {\em flow}. Given a graph $G$, we view a flow in $G$ as a vector $\ff\in \bbR^m$ that assigns a value $f_e$ to each  arc $e$ of $G$. If this value is negative we interpret it as having a flow of $|f_e|$ flowing in the direction opposite to the arc orientation. (This convention is especially useful when discussing flows in undirected graphs.)

We say that a flow $\ff$ is an {\em $\vsigma$-flow}, for some {\em demands} $\vsigma\in \bbR^n$ iff it satisfies {\em flow conservation constraints} with respect to that demands. That is, we have that
\begin{equation}\label{eq:conservation_constraints}
\sum_{e\in E^+(v)} f_{e} - \sum_{e\in E^-(v)} f_{e} = \sigma_v, \quad \text{for each vertex $v\in V$}.
\end{equation} 
Here, $E^+(v)$ (resp. $E^-(v)$) is the set of arcs of $G$ that are entering (resp. leaving) vertex $v$. Intuitively, these constraints enforce that the net balance of the total in-flow into vertex $v$ and the total out-flow out of that vertex is equal to $\sigma_v$, for every $v\in V$.  (Observe that this implies, in particular, that $\sum_v \sigma_v =0$.)

Furthermore, we say that a $\vsigma$-flow $\ff$ is {\em feasible} in $G$ iff $\ff$ obeys the {\em the capacity constraints}:
\begin{equation}\label{eq:capacity_constraints}
-u_e^-\leq f_e \leq u_e^+, \quad \text{for each arc $e\in E$}.
\end{equation} 
In other words, we want each arc $e$ to have a flow that is at most $u_e^+$ if it flows in the direction of $e$'s orientation (i.e., $f_e\geq 0$), and at most $u_e^-$, if it flows in the opposite direction (i.e., $f_e<0$). Note that setting all $u_e^-$s be equal to zero recovers the standard notion of flow feasibility in directed graphs. 

 One type of flows that will be of special interest to us are $s$-$t$ flows, where $s$ (the {\em source}) and $t$ (the {\em sink}) are two distinguish vertices of $G$. Formally, an {\em $s$-$t$ flow} is a $\vsigma$-flow whose demand vector $\vsigma$ is equal to $F\cdot \vchi_{s,t}$, where $F\geq 0$ is called the {\em value} of $\ff$ and $\vchi_{s,t}$ is a demand vector that has $-1$ (resp. $1$) at the coordinate corresponding to $s$ (resp. $t$) and zeros everywhere else.

Now, the {\em maximum flow problem} corresponds to a task in which we are given a (directed) graph $G=(V,E, \uu)$ with integer capacities as well as a source vertex $s$ and a sink vertex $t$ and want to find a {\em feasible} (in the sense of \eqref{eq:capacity_constraints}) $s$-$t$ flow of maximum value. We will denote this maximum value as $F^*$.

\paragraph{Residual Graphs.} A fundamental object in many maximum flow algorithms (including ours) is the notion of a residual graph. Given a graph $G=(V,E,\uu)$ and a feasible $\vsigma$-flow $\ff$ in that graph (it is useful to think $\vsigma=F\cdot \vchi_{s,t}$), we define the {\em residual graph $G_{\ff}$} (of $G$ with respect to $\ff$) as a graph $G_{\ff}=(V,E,\huu(\ff))$ over the same vertex and arc set as $G$ and such that, for each arc $e=(u,v)$ of $G$, its lower and upper capacities are defined as
\begin{equation}
\label{eq:capacities_residual_graph}
\hu_e^+(\ff):=u_e^+-f_e \text{\ \  and \ \ } \hu_e^-(\ff):=u_e^-+f_e.
\end{equation}
We will refer to $\hu_e^+(\ff)$ (resp. $\hu_e^-(\ff)$) as {\em forward residual capacity} (resp. {\em backward residual capacity}) of $e$ and also define the {\em residual capacity} $\hu_e(\ff)$ of $e$ as the minimum of these two, i.e., $\hu_e(\ff):=\min\{\hu_e^-(\ff),\hu_e^+(\ff)\}$.  Note that the value of residual capacity depends on the flow $\ff$ but we will ensure that it is always clear from the context with respect to which flow the residual capacity is measured. Also, observe that feasibility of $\ff$ implies that all residual capacities are always non-negative (cf. \eqref{eq:capacity_constraints}).

The main reason why residual graphs are useful in computing maximum flows is that they constitute a very convenient representation of the progress made so far. Specifically, we have the following important fact. (Again, it is useful to think here of the maximum $s$-$t$ flow setting, in which  $\vsigma=F^*\vchi_{s,t}$.)

\begin{fact}
\label{fa:residual_graph}
Let $\vsigma$ be some demand and $G=(V,E,\uu)$ be a graph in which a demand of $\vsigma$ can be routed, i.e., there exists a $\vsigma$-flow $\ff^*$ that is feasible in $G$. Also, for any $0\leq \alpha \leq 1$, let $\ff$ be a feasible $\alpha\vsigma$-flow in $G$, and $G_{\ff}=(V,E,\huu(\ff))$ be the residual graph of $G$ with respect to $\ff$. We have that
\begin{enumerate}[(a)]
\item one can route a demand of $(1-\alpha)\vsigma$ in $G_{\ff}$;
\item if $\ff'$ is a feasible $\alpha'\vsigma$-flow in $G_{\ff}$, for some $\alpha'$, then $\ff+\ff'$ is a feasible $(\alpha+\alpha')\vsigma$-flow in $G$.
\end{enumerate}
\end{fact}

Intuitively, the above fact enables us to reduce the task of routing a demand $\vsigma$ in $G$ to a sequence of computations of augmenting $\alpha'\vsigma$-flows in the residual graph $G_{\ff}$. We know that as long as we have not yet computed a feasible $\vsigma$-flow in $G$, $G_{\ff}$ can route a demand of $(1-\alpha)\vsigma$-flow, where $(1-\alpha)>0$ is the fraction of routed demand that we are still ``missing'', and each new augmenting $\alpha'\vsigma$-flow found in $G_{\ff}$ brings us closer to routing $\vsigma$ in full in $G$. (Note that initially $G_{\ff}$ is equal to $G$ and $G_{\ff}$ is changing after each new augmenting $\alpha'\vsigma$-flow is found.)  

\paragraph{Electrical Flows and Vertex Potentials.}

Another notion that will play a fundamental role in this paper is the notion of electrical flows. Here, we just briefly review some of the key properties that we will need later. For an in-depth treatment we refer the reader to \cite{Bollobas98}. 

Consider a graph $G$ and a vector of resistances $\rr\in \bbR^{m}$ that assigns to each edge $e$ its {\em resistance} $r_e>0$. For a given $\vsigma$-flow $\ff$ in $G$, let us define its {\em energy} (with respect to resistances $\rr$) $\energy{\rr}{\ff}$ to be
\begin{equation}\label{eq:def_energy_flow}
\energy{\rr}{\ff}:= \sum_e r_e f_e^2.
\end{equation}

For a given demand vector $\vsigma$ and a vector of resistances $\rr$, we define the {\em electrical $\vsigma$-flow} in $G$ (that is {\em determined} by resistances $\rr$) to be the flow that minimizes the energy $\energy{\rr}{\ff}$ among all flows with demand $\vsigma$ in $G$. As energy is a strictly convex function, one can easily see that such a flow is unique. (It is important to keep in mind that such flow is \emph{not} required to be feasible with respect to capacities of $G$, in the sense of \eqref{eq:capacity_constraints}.)

A very useful property of electrical flows is that they can be characterized in terms of vertex potentials inducing them. Namely, one can show that a flow $\ff$ with demands $\vsigma$  in $G$ is an electrical $\vsigma$-flow determined by resistances $\rr$ iff there exist {\em vertex potentials} $\phi_v$ (that we collect into a vector $\vphi\in \bbR^n$) such that, for any edge $e=(u,v)$ in $G$,
\begin{equation}\label{eq:Ohms_lawy}
f_e = \frac{\phi_v-\phi_u}{r_e}.
\end{equation}
In other words, a $\ff$ with demands $\vsigma$ is an electrical $\vsigma$-flow iff it is {\em induced} via \eqref{eq:Ohms_lawy} by some vertex potential $\vphi$. (Note that the orientation of edges matters in this definition.) The above equation corresponds to the Ohm's law known from physics.

Note that we are able to express the energy $\energy{\rr}{\ff}$ (see \eqref{eq:def_energy_flow}) of an electrical $\vsigma$-flow $\ff$ in terms of the potentials $\vphi$ inducing it as
\begin{equation}\label{eq:def_energy_potentials}
\energy{\rr}{\ff}= \sum_{e=(u,v)} \frac{(\phi_v-\phi_u)^2}{r_e}.
\end{equation}

One of the consequences of the above is that one can develop a dual characterization of the energy of an electrical $\vsigma$-flow in terms of optimization over vertex potentials. Namely, we have the following lemma whose proof can be found, e.g., in \cite{Madry13} Lemma 2.1.


\begin{lemma}\label{lem:effective_conductance}
For any graph $G=(V,E)$, any vector of resistances $\rr$, and any demand vector $\vsigma$, 
\[
\frac{1}{\energy{\rr}{\ff^*}} = \min_{\vphi| \vsigma^T \vphi=1} \sum_{e=(u,v)\in E} \frac{(\phi_v-\phi_u)^2}{r_{e}},
\]
where $\ff^*$ is the electrical $\vsigma$-flow determined by $\rr$ in $G$. Furthermore, if $\vphi^*$ are the vertex potentials corresponding to $\ff^*$ then the minimum is attained by taking $\vphi$ to be equal to $\hvphi:=\vphi^*/\energy{\rr}{\ff^*}$.
\end{lemma}

Note that the above lemma provides a convenient way of lowerbounding the energy of an electrical $\vsigma$-flow. One just needs to expose any vertex potentials $\vphi$ such that $\vsigma^T \vphi=1$ and this will immediately constitute an energy lowerbound. 

%

\paragraph{Laplacian Solvers.} The fact that the electrical $\vsigma$-flow determined by resistances $\rr$ is the only flow with demands $\vsigma$ that can be induced by vertex potentials (cf. \eqref{eq:Ohms_lawy}) has an important consequence. It enables us to reduce electrical $\vsigma$-flow computations to solving a linear system. In fact, the task of finding vertex potentials that induce that flow can be cast as a {\em Laplacian} linear system. That is, a linear system in which the constraint matrix corresponds to a Laplacian of the underlying graph with weights given by the (inverses of) the resistances $\rr$.

Now, from the algorithmic point of view, the crucial property of Laplacian systems is that we can solve them, up to a very good approximation, very efficiently. Namely, there is a long line of work \cite{SpielmanT04,KoutisMP10,KoutisMP11,KelnerOSZ13,CohenKMPPRX14,KyngLPSS16,KyngS16} that builds on an earlier work of Vaidya \cite{Vaidya90} and Spielman and Teng \cite{SpielmanT03} and gives us a number of Laplacian system solvers that run in only nearly-linear time and, in case of more recent variants, are conceptually fairly simple. In particular, this line of work establishes the following theorem.

\begin{theorem}\label{thm:vanilla_SDD_solver}
For any $\eps>0$, any graph $G$ with $n$ vertices and $m$ edges, any demand vector $\vsigma$, and any resistances $\rr$, one can compute in $\tO{m \log \eps^{-1}}$ time vertex potentials $\tvphi$ such that
$\|\tvphi-\vphi^*\|_{\LL}\leq \eps \|\vphi^*\|_{\LL}$, where $\LL$ is the Laplacian of $G$, $\vphi^*$ are potentials inducing the electrical $\vsigma$-flow determined by resistances $\rr$, and $\|\vphi\|_{\LL}:=\sqrt{\vphi^T \LL \vphi}$.
\end{theorem}

Even though the solutions delivered by the above Laplacian solvers are only approximate, the quality of approximation that it delivers is more than sufficient for our purposes. Therefore, in the rest of this paper we assume that these solutions are exact. (See, e.g., \cite{Madry13} for discussion how to deal with inexactness of the electrical flows computed.)

\paragraph{Vector Norms.} We will find it useful to work with various $\ell_p$-norms of vectors. To this end, for any $p>0$, we define the {\em $\ell_p$-norm} $\norm{\hh}{p}$ of a vector $\hh$ as 
\begin{equation}
\label{eq:l_p_norms}
\norm{\hh}{p}:=\left(\sum_i |h_i|^p\right)^{\frac{1}{p}}.
\end{equation}
In particular, the $\ell_{\infty}$-norm is given by $\inorm{\hh}:=\max_i |h_i|$. 


\section{Augmenting Residual Graphs with Electrical Flows}\label{sec:basic_algorithm}

In this section, we put forth the general framework we will use to solve the maximum $s$-$t$ problem.\footnote{Observe that our treatment can be immediately extended to routing arbitrary demands since this problem can always be reduced to the directed maximum flow problem by adding a super-source and super-sink and connecting all the surpluses in the original graph to the former and all the deficits in that graph to the latter via arcs with appropriate capacities.} In particular, we demonstrate how this framework enables us to solve the maximum $s$-$t$ flow problem in $\tO{m^{\frac{3}{2}}\log U}$ time, where $m=|E|$ is the number of arcs in of the input graph and $U$ is its largest (integer) capacity. 

More precisely, for any maximum $s$-$t$ flow instance $G=(V,E,\uu)$ and any value $F\geq 0$, our algorithm will work in $\tO{m^{\frac{3}{2}}\log U}$ time and either: compute a feasible $s$-$t$ flow of value $F$ in $G$; or conclude that the maximum $s$-$t$ flow value $F^*$ of $G$ is strictly smaller than $F$. 

Note that such procedure can be turned into a ``classic'' maximum $s$-$t$ flow by applying binary search over values of $F$ and incurring a multiplicative running time overhead of only $O(\log Un)$. (In fact, a standard use of capacity scaling  technique \cite{EdmondsK72} enables one to keep the overall running time of the resulting algorithm be only linear, instead of quadratic, in $\log U$.)

%
%

As mentioned earlier, our algorithm follows the primal dual augmenting paths based framework. At a high level, in each iteration (see Section \ref{sec:progress_steps}), it uses electrical flow computations to compute an augmenting flow as well as an update to the dual solution. To ensure that each augmenting iteration makes enough progress, it maintains a careful coupling of the primal and dual solution. We describe it below.

\subsection{Primal Dual Coupling}\label{sec:primal_dual_coupling}

Let us fix our target flow value $F$ and, for notational convenience, , for any $0\leq \alpha \leq 1$, let us denote by $\vchi_{\alpha}$ the demand vector $\alpha F\vchi_{s,t}$, i.e., the demand corresponding to routing $\alpha$-fraction of the target flow value $F$ of the $s$-$t$ flow. Also, let us define $\vchi$ to be the demand equal to $\vchi_1$.

Again, our algorithm will be inherently primal dual in nature. That is, in addition to maintaining a primal solution: a $\vchi_\alpha$-flow $\ff$, for some $0\leq \alpha\leq 1$, that is feasible in $G$, it will also maintain a dual solution $\yy\in \bbR^{n}$, which should be viewed as an embedding of all the vertices of the graph $G$ into a line. 

Consequently, our goal will be to either to make $\ff$ be a feasible flow with demands $\vchi_{\alpha}$ and $\alpha=1$, which corresponds to routing the target $s$-$t$ flow value $F$ in full, or to make the dual solution $\yy$ certify that the target demand $\vchi_1=\vchi$ cannot be fully routed in $G$ and thus $F>F^*$. 

\paragraph{Well-coupled Solutions.} Our primal dual scheme will be enforcing a very specific coupling of the primal solution $\ff$ and the dual solution $\yy$. More precisely, let us define for each arc $e=(u,v)$ 
\begin{equation}\label{eq:def_delta_y}
\Delta_e( \yy):=y_v-y_u,
\end{equation}
to be the ``stretch`` of the arc $e$ in the embedding given by $\yy$. Also, let $G_{\ff}$ be the residual graph of $G$ with respect to $\ff$ and let us define, for a given arc $e$, a potential function
\begin{equation}\label{eq:def_phi_e}
\Phi_e(\ff):=\frac{1}{\hu_e^+(\ff)}-\frac{1}{\hu_e^-(\ff)},
\end{equation}
where we recall that $\hu_e^+(\ff):=u_e^+-f_e$ (resp. $\hu_e^-(\ff):=u_e^-+f_e$) are forward (resp. backward) residual capacities of the arc $e$. (See preliminaries, i.e., Section \ref{sec:preliminaries}, for details.)
 
 Then, our intention is to maintain the following relation between $\ff$ and $\yy$:
\begin{equation}\label{eq:ideal_coupling}
\Delta_e( \yy)= \Phi_e(\ff) \text{ \ \ \ for each arc $e$}.
\end{equation}

Intuitively, this condition ensures that the line embedding $\yy$ stretches each arc $e$ in the direction of the smaller of the residual capacities, and that this stretch is inversely proportional to the value of that capacity. (Note that if $G$ was undirected and thus the initial capacities $u_e^+$ and $u_e^-$ were equal, the direction of smaller residual capacity is also the direction in which the flow $\ff$ flows through the arc $e$.) It is worth pointing out that this coupling condition is directly inspired by (and, in fact, can be directly derived from) a certain variant of centrality condition used by interior-point method based maximum flow algorithms (see \cite{Madry13}). 

Even though condition \eqref{eq:ideal_coupling} expresses our intended coupling, it will be more convenient to work with a slightly relaxed condition that allows us to have small violations of that ideal coupling. Specifically, we say that a primal dual solution $(\ff, \yy)$ is {\em $\vgamma$-coupled} iff
\begin{equation}\label{eq:coupling}
\left|\Delta_e( \yy)-\Phi_e(\ff)\right| \leq \frac{\gamma_e}{\hu_e(\ff)}\text{ \ \ \  for all arcs $e=(u,v)$}.
\end{equation}
Here, $\hu_e(\ff)=\min\{\hu_e^+(\ff),\hu_e^-(\ff)\}$ is the (symmetrized) residual capacity of the arc $e$, and $\vgamma \in \bbR^m$ is the {\em violation vector} that we intend to keep very small. In particular, we say that a primal dual solution $(\ff,\yy)$ is {\em well-coupled} iff its violation vector $\vgamma$ has its $\ell_2$-norm (see \eqref{eq:l_p_norms}) be at most $\frac{1}{100}$, i.e., $\|\vgamma\|_2\leq \frac{1}{100}$.

One of the key consequences of maintaining a well-coupled primal dual pair of solutions $(\ff,\yy)$  is that it enables us to use $\yy$ as a dual certificate for inability to route certain demands in $G$. The following lemma gives us a concrete criterion for doing that.

\begin{lemma}\label{lem:duality_certificate}
Let $(\ff,\yy)$ be a well-coupled primal dual solution with $\ff$ being a $\vchi_\alpha$-flow, for some $0\leq \alpha <1$. If 
\[
\vchi^T\yy > \frac{2m}{(1-\alpha)}
\]
then the demand $\vchi$ cannot be routed in $G$, i.e., $F>F^*$.
\end{lemma}

Note that the choice of the constant $2$ in the above lemma is fairly arbitrary. In principle, any constant strictly larger than $1$ would suffice.

%

\begin{proof}
Assume for the sake of contradiction that $\vchi^T\yy > \frac{2m}{(1-\alpha)}$ but the demand $\vchi$ can still be routed in $G$. As $\ff$ is an $\vchi_{\alpha}$-flow that is feasible in $G$, by Fact \ref{fa:residual_graph}, there has to exist a $\vchi_{(1-\alpha)}$-flow $\ff'$ that is feasible in the residual graph $G_{\ff}$. Observe that, by \eqref{eq:conservation_constraints} and \eqref{eq:def_delta_y}, we have in that case that
\begin{eqnarray}\label{eq:ff_prime_inner}
(\ff')^T \Delta(\yy) & = & \sum_{e} f_e' \Delta_e( \yy)= \sum_{e=(u,v)} f_e' \left(y_v-y_u\right) = \sum_{v} y_v \left(\sum_{e\in E^+(v)} f_e' - \sum_{e\in E^-(v)} f_e'\right)\\
&= & (1-\alpha) F(y_t-y_s) = (1-\alpha) \vchi^T \yy > 2m.\nonumber
\end{eqnarray}

We want now to use the feasibility of $\ff'$ to derive an upper bound on the inner produce $(\ff')^T \Delta(\yy)$ that will violate the above lower bound and thus deliver the desired contradiction. 

To this end, let us consider some particular arc $e$. Assume wlog that $f_e'\geq 0$ (handling of the case of $f_e'<0$ is completely analogous). We have that the contribution of that arc to the inner product $(\ff')^T\Delta(\yy)$ is, by \eqref{eq:coupling}, at most
\[
f_e'\Delta_e(\yy) \leq f_e' \Phi_e(\ff)+\frac{\gamma_e f_e'}{\hu_e(\ff)} \leq f_e' \left(\frac{1}{\hu_e^+(\ff)}-\frac{1}{\hu_e^-(\ff)}\right)+\frac{\gamma_e f_e'}{\hu_e(\ff)},
\]
where we recall that $\vgamma$ is the violation vector of $(\ff, \yy)$ (cf. \eqref{eq:coupling}).

As $\ff'$ is feasible in $G_{\ff}$, we need to have that $f_e'\leq \hu_e^+(\ff)$. So, if $\hu_e^+(\ff)\leq \hu_e^-(\ff)$, i.e., $\hu_{e}^+(\ff)=\hu_e(\ff)$, then 
\[
f_e'\Delta_e(\yy) \leq f_e' \left(\frac{1}{\hu_e^+(\ff)}-\frac{1}{\hu_e^-(\ff)}\right)+\frac{\gamma_e f_e'}{\hu_e(\ff)} \leq \frac{f_e'(1+\gamma_e)}{\hu_e(\ff)}\leq (1+\gamma_e).
\]
Otherwise, we have that
\[
f_e'\Delta_e(\yy) \leq f_e' \left(\frac{1}{\hu_e^+(\ff)}-\frac{1}{\hu_e^-(\ff)}\right)+\frac{\gamma_e f_e'}{\hu_e(\ff)} \leq f_e' \left(\frac{1}{\hu_e^+(\ff)}-\frac{1-\gamma_e}{\hu_e(\ff)}\right)\leq 1-\frac{(1-\gamma_e)f_e'}{\hu_e(\ff)}\leq 1,
\]
as $\gamma_e\leq \inorm{\vgamma} \leq \norm{\vgamma}{2}\leq \frac{1}{2}$.  So, in either case, this contribution is at most $(1+\gamma_e)$. Summing these contributions over all arcs we get that
\[
(\ff')^T\Delta(\yy) \leq \sum_{e} f_e' \Delta_e( \yy) \leq \sum_e (1+\gamma_e)\leq m + \sqrt{m} \norm{\vgamma}{2} \leq 2 m,
\]
which indeed contradicts the lower bound \eqref{eq:ff_prime_inner}. The lemma follows.
\end{proof}

In the light of the above discussion, for a given well-coupled primal dual solution $(\ff,\yy)$, we should view the value of $\alpha$ as a measure of our primal progress, while the value of the inner product $\vsigma^T\yy$ can be seen as a measure of our dual progress.

\paragraph{Initialization.} The coupling condition \eqref{eq:coupling} ties the primal and dual solutions $\ff$ and $\yy$ fairly tightly. In fact, coming up with some primal dual solutions that are well-coupled, which we need to initialize our framework,  might be difficult.

Fortunately, finding such a pair of initial well-coupled solutions turns out to be easy, if our input graph $G$ is undirected. In that case,  just taking $\ff$ to be a trivial zero flow and $\yy$ to be a trivial all-zeros embedding makes condition \eqref{eq:coupling} satisfied (with $\gamma_e$s being all zero). After all, the residual graph $G_{\ff}$ with respect to such zero flow $\ff$ is just the graph $G$ itself and thus $\hu_e^+(\ff)=u_e^+=u_e^-=\hu_e^-(\ff)$. 

Furthermore, even though we are interested in solving directed instances too, every such instance can be reduced to an undirected one. Specifically, we have the following lemma.

\begin{lemma}\label{lem:initialization}
Let $G$ be an instance of the maximum $s$-$t$ flow problem with $m$ arcs and the maximum capacity $U$, and let $F$ be the corresponding target flow value. In $\tO{m}$ time, one can construct an instance $G'$ of  undirected maximum $s$-$t$ flow problem that has $O(m)$ arcs and the maximum capacity $U$, as well as target flow value $F'$ such that:
\begin{enumerate}[(a)]
\item if there exists a feasible $s$-$t$ flow of value $F$ in $G$ then a feasible $s$-$t$ flow of value $F'$ exists in $G'$;
\item given a feasible $s$-$t$ flow of value $F'$ in $G'$ one can construct in $\tO{m}$ time a feasible $s$-$t$ flow of value $F$ in $G$.
\end{enumerate}
\end{lemma}

The proof of the above lemma boils down to a known reduction of the directed maximum flow problem to its undirected version -- see, e.g., Theorem 3.6.1 in \cite{Madry11} for details. Consequently, from now on we can assume without loss of generality that we always have a well-coupled primal dual pair to initialize our framework.

\subsection{Progress Steps}\label{sec:progress_steps}

Once we described our basic framework and how to initialize it, we are ready to put forth its main ingredient: progress steps that enable us to gradually improve the primal dual solutions that we maintain. To this end, let us fix a well-coupled primal dual solutions $(\ff,\yy)$ with $\ff$ being a $\vchi_\alpha$-flow, for some $0\leq \alpha<1$, that is feasible in $G$. Our goal in this section will be to use  $(\ff,\yy)$ to compute, in nearly-linear time, another pair of well-coupled primal dual solutions $(\ff^+,\yy^+)$ that bring us closer to the optimal solutions. The flow $\ff^+$ we obtain will be a $\vchi_{\alpha'}$-flow feasible in $G$, for $\alpha'>\alpha$. So, the resulting flow update $\ff^+-\ff$ is an augmenting flow that is feasible in our current residual graph $G_{\ff}$ and pushes $(\alpha'-\alpha)$-fraction of the target $s$-$t$ flow. 

We will compute $(\ff^+,\yy^+)$ in two stages. First, in the {\em augmentation step}, we obtain a pair of solutions $(\hff, \hyy)$, with $\hff$ being a $\vchi_{\alpha'}$-flow, for $\alpha'>\alpha$, that is feasible in $G$. These solutions make progress toward the optimal solutions but might end up being not well-coupled. Then, in the {\em fixing step}, we correct $(\hff, \hyy)$ slightly by adding a carefully chosen flow circulation, i.e., a flow with all-zeros demands, to $\hff$ and an dual update to $\hyy$ so as to make the resulting solutions $(\ff^+,\yy^+)$ be well-coupled, as desired.

The key primitive in both these steps will be electrical flow computation. As we will see, the crucial property of electrical flows we will rely on here is their``self-duality``. That is, the fact that each electrical flow computation gives us both the flow and the corresponding vertex potentials that are coupled to it via Ohm's law \eqref{eq:Ohms_lawy}. This enables us not only to update our primal and dual solutions with that flow and vertex potentials, respectively, but also, much more crucially,  this coupling introduced by Ohm's law will be exactly what will allow us to (approximately) maintain our desired primal dual coupling property \eqref{eq:coupling}.

\paragraph{Augmentation Step.} To perform an augmentation step we compute first an electrical $\vchi$-flow $\tff$ in $G$ with the resistances $\rr$ defined as
\begin{equation}\label{eq:res_definition}
r_e:=\frac{1}{(\hu_e^+(\ff))^2}+\frac{1}{(\hu_e^-(\ff))^2},
\end{equation}
for each arc $e$. Note that the resistance $r_e$ is proportional, roughly, to the inverse of the square of the residual capacity $\hu_e(\ff)$ of that arc. So, in particular, it becomes very large whenever residual capacity of the arc $e$ is small, and vice versa. As we will see shortly, this correspondence will allow us to control the amount of flow that $\tff$ sends over each edge and thus ensure that the respective residual capacities are not violated. 

Let $\tvphi$ be the vertex potentials inducing $\tff$ (via the Ohm's law \eqref{eq:def_energy_potentials}). Then, we obtain the new primal and dual solution $(\hff, \hyy)$ as follows:
\begin{eqnarray}\label{eq:augmentation_step_update}
\hf_e & := & f_e + \delta \tf_e \text{ \ \ \ \ for each arc $e$}\\
\hy_v & := & y_v + \delta \tphi_v \text{ \ \ \ \ for each vertex $v$},\nonumber
\end{eqnarray}
where $\delta$ is the desired {\em step size}. Observe that this update is exactly an  augmentation of our current flow $\ff$ with the (scaled) electrical flow $\delta \tff$ and adding to our dual solution the (scaled) vertex potentials $\delta \tvphi$. This, in particular, means that the new flow $\hff$ we obtain here is a $\vchi_{\alpha'}$-flow with 
\begin{equation}\label{eq:growth_alpha}
\alpha'=\alpha + \delta.
\end{equation}

The step size $\delta$, however, will have to be carefully chosen. On one hand, as we see in \eqref{eq:growth_alpha}, the larger it is the more progress we make. On the other hand, though, it has to be small enough so as to keep the flow $\delta\tff$ feasible in $G_{\ff}$ (and thus, by Fact \ref{fa:residual_graph}, to make the flow $\hff+\ff$ feasible in $G$).

Note that, a priori, we have no direct control over neither the directions in which the electrical $\vchi$-flow $\tff$ is flowing thorough each arc nor the amount of that flow. So, in order to establish a grasp on what is the right setting of $\delta$, it is useful to define a {\em congestion vector} $\vrho$ given by
\begin{equation}
\label{eq:congestion_vec_def}
\rho_e :=\frac{\tf_e}{\hu_e(\ff)},
\end{equation} 
for each arc $e$. One can view $\rho_e$ as a normalized measure of how much the electrical flow $\tff$ overflows the residual capacity $\hu_e(\ff)$ of the arc $e$ and in what direction. In other words, the sign of $\rho_e$  encodes the direction of the flow $\tf_e$.

It is now not hard to see that to ensure that $\delta\tff$ is feasible in $G_{\ff}$, i.e., that no residual capacity is violated by the update \eqref{eq:augmentation_step_update}, it suffices that $\delta |\rho_e|\leq \frac{1}{4}$, for all arcs $e$, or, equivalently, that
\begin{equation}
\label{eq:l_infty_bound}
\delta \leq \frac{1}{4\inorm{\vrho}},
\end{equation} 
where $\inorm{\cdot}$ is the standard $\ell_\infty$-norm.

It is also worth pointing out that the congestion vector $\vrho$ turns out to capture (up to a small multiplicative factor) the contribution of each arc $e$ to the energy $\energy{\rr}{\tff}$ of $\tff$. In particular, we have the following simple but important observation.
\begin{lemma}\label{lem:energy_rho}
For any arc $e$, $\rho_e^2 \leq r_e \tf_e^2 \leq 2\rho_e^2$ and
\[
\norm{\vrho}{2}^2 \leq \energy{\rr}{\tff} \leq 2 \norm{\vrho}{2}^2,
\]
where $\norm{\cdot}{2}$ is the standard $\ell_2$-norm.
\end{lemma}
\begin{proof} Note that, by definition of the resistance $r_e$ \eqref{eq:res_definition}, \eqref{eq:congestion_vec_def}, and the definition of the residual capacity $\hu_e$, we have that
\[
\rho_e^2 \leq \frac{(\rho_e \hu_e(\ff))^2}{\hu_e(\ff)^2} \leq \frac{\tf_e^2}{\hu_e(\ff)^2} \leq \left((\hu_e^+(\ff))^{-2}+(\hu_e^-(\ff))^{-2}\right) \tf_e^2 = r_e \tf_e^2 \leq \frac{2\tf_e^2}{\hu_e(\ff)^2} = 2\rho_e^2.
\]
We thus also have that
\[
\norm{\vrho}{2}^2=\sum_e \rho_e^2 \leq \sum_e r_e \tf_e^2 = \energy{\rr}{\tff} \leq \sum_e 2 \rho_e^2 = 2 \norm{\vrho}{2}^2,
\]
as desired.
\end{proof}

This link between the energy-minimizing nature of the electrical $\vsigma$-flow $\tff$ and  the $\ell_2$-norm of the congestion vector $\vrho$ will end up being very important. One reason for that is the fact that $\ell_\infty$-norm is always bounded by the $\ell_2$-norm. Consequently, we can use this connection to control the $\ell_\infty$-norm of the vector $\vrho$ and thus the value of $\delta$ needed to satisfy the feasibility condition \eqref{eq:l_infty_bound}.

It turns out, however, that just ensuring that our augmenting flow is feasible is not enough for our purposes. Specifically, we also need to control the coupling of our primal dual solutions, and the feasibility bound \eqref{eq:l_infty_bound} might be not sufficiently strong for this purpose. We thus have to develop analyze the impact of the update \eqref{eq:augmentation_step_update} on the coupling condition \eqref{eq:coupling} more closely.

To this end, let us first notice the following fact that stems from a standard application of the Taylor's theorem. Its proof appears in Appendix \ref{app:Taylor_approximation}

\begin{fact}
\label{fa:Taylor_approximation}
For any $u_1,u_2>0$ and $x$ such as $|x|\leq \frac{u}{4}$, where $u=\min\{u_1,u_2\}$, we have that
\[
\left(\frac{1}{u_1-x}-\frac{1}{u_2+x}\right)=\frac{1}{u_1}-\frac{1}{u_2}+\left(\frac{1}{u_1^2}+\frac{1}{u_2^2}\right)x+x^2\zeta,
\]
where $|\zeta|\leq \frac{5}{u^3}$.
\end{fact}

Now, the above approximation bound enables us to get an accurate estimate of how the coupling condition evolves during the augmentation step \eqref{eq:augmentation_step_update}. That is, for any arc $e$, applying Fact \ref{fa:Taylor_approximation} with $u_1=\hu_e^+(\ff)$, $u_2=\hu_e^-(\ff)$ and $x=\delta \tf_e$, gives us that
\begin{equation}\label{eq:Taylor_approx}
\Phi_e(\hff)=\frac{1}{\hu_e^+(\ff)-\delta\tf_e}-\frac{1}{\hu_e^-(\ff)+\delta\tf_e} = \Phi_e(\ff)+\left(\frac{1}{(\hu_e^+(\ff))^2}+\frac{1}{(\hu_e^-(\ff))^2}\right)\delta \tf_e +(\delta \tf_e)^2\zeta_e, 
\end{equation}
with $|\zeta_e| \leq \frac{5}{\hu_e(\ff)^3}$.

Observe that the above expression tells us that the first order approximation of the change in the primal contribution of the arc $e$ to the coupling condition \eqref{eq:coupling} caused by the update \eqref{eq:augmentation_step_update} is exactly
\[
\Phi_e(\hff)-\Phi_e(\ff)\approx \left(\frac{1}{(\hu_e^+(\ff))^2}+\frac{1}{(\hu_e^-(\ff))^2}\right)\delta \tf_e = r_e \delta \tf_e,
\]
where we also used \eqref{eq:res_definition}. (In fact, the choice of the resistances $\rr$ was made exactly to make the above statement true.) 

Furthermore, by Ohm's law \eqref{eq:Ohms_lawy} and the definition of our augmentation step \eqref{eq:augmentation_step_update}, we have that
\[
r_e \delta \tf_e = \delta \left(\tphi_v-\tphi_u \right)=\Delta_e( \hyy)-\Delta_e( \yy),
\]
which is exactly the change in the dual contribution of the arc $e=(u,v)$ to the coupling condition \eqref{eq:coupling} caused by the augmentation step update \eqref{eq:augmentation_step_update}. 

So, up to first order approximation, these two contributions cancel out, leaving the coupling \eqref{eq:coupling} intact. Consequently, any increase in the violation of the coupling condition must come from the second-order terms in the approximation \eqref{eq:Taylor_approx}. The following lemma makes this precise and is proved in Appendix \ref{app:augmentation_coupling}.

\begin{lemma}
\label{lem:augmentation_coupling}
Let $0<\delta \leq (4\inorm{\vrho})^{-1}$ and the primal dual solution $(\ff,\yy)$ be $\vgamma$-coupled. Then, we have that, for any arc $e=(u,v)$,
\[
\left|\Delta_e(\hyy)-\Phi_e(\hff)\right|\leq \frac{\frac{4}{3}\gamma_e + 7(\delta\rho_e)^2}{\hu_e(\hff)}.
\]
\end{lemma}

\paragraph{Fixing Step.} Although Lemma \ref{lem:augmentation_coupling} enables us to bound the deterioration of the primal dual coupling during the augmentation step, we cannot prevent this effect altogether. Therefore, we need to introduce a {\em fixing step} that deals with this problem. More precisely, we develop a procedure that uses a single electrical flow computation to significantly reduce that violation, provided it was not too large to begin with. This is formalized by the following lemma, whose proof appears in Appendix \ref{app:fixing_step}.

\begin{lemma}
\label{lem:fixing_step}
Let $(\gg,\zz)$ be a $\vvarsigma$-coupled primal dual solution, with $\gg$ being a feasible $\vchi_{\alpha'}$-flow and $\norm{\vvarsigma}{2}\leq \frac{1}{50}$. In $\tO{m}$ time, we can compute a primal dual solution $(\ogg, \ozz)$ that is well-coupled and in which $\ogg$ is still a $\vchi_{\alpha'}$-flow.
\end{lemma}

Now, after putting Lemmas \ref{lem:augmentation_coupling} and \ref{lem:fixing_step} together, we are finally able  to state the condition that $\delta$ in the update \eqref{eq:augmentation_step_update} has to satisfy in order to ensure that the solutions $(\ff^+,\yy^+)$ we obtain after performing the augmentation and fixing step is still well-coupled.  

\begin{lemma}\label{lem:l_4_condition_step}
$(\ff^+, \yy^+)$ is a well-coupled primal dual solution with $\ff^+$ being a $\vchi_{\alpha'}$-flow that is feasible in $G$ whenever 
\[
\delta \leq \left(33 \norm{\vrho}{4}\right)^{-1},
\]
\end{lemma}

The above lemma tells us that the step size $\delta$ of our augmentation step \eqref{eq:augmentation_step_update} should be governed by the $\ell_4$-norm (see \eqref{eq:l_p_norms}) of the congestion vector \eqref{eq:congestion_vec_def}. Observe that the $\ell_4$-norm of a vector is always upper bounding its $\ell_\infty$-norm. So, the condition \eqref{eq:l_infty_bound} is subsumed by this $\ell_4$-norm bound.

\begin{proof}
Note that we always have that $\norm{\vrho}{4}\geq \norm{\vrho}{\infty}$. So, the condition \eqref{eq:l_infty_bound} is automatically satisfied and the flow $\ff^+$ is indeed a $\vchi_{\alpha'}$-flow that is feasible in $G$.

Now, to argue about well-coupling of $(\ff^+, \yy^+)$, notice that, in the light of Lemma \ref{lem:fixing_step}, it suffices to argue that the primal dual solution $(\hff, \hyy)$ obtained after executing the augmentation step \eqref{eq:augmentation_step_update} is $\hvgamma$-coupled with $\norm{\hvgamma}{2}\leq \frac{1}{50}$. 

To this end, observe that by Lemma \ref{lem:augmentation_coupling} and the fact that $(\ff, \yy)$ are well-coupled, i.e., $\vgamma$-coupled with $\norm{\vgamma}{2}\leq \frac{1}{100}$, we have that
\begin{eqnarray*}
\norm{\hvgamma}{2} &= &\sqrt{\sum_e \hgamma_e^2} \leq \sqrt{ \sum_e \left(\frac{4}{3}\gamma_e + 7 (\delta \rho_e)^2\right)^2 }\\
&\leq & \frac{4}{3}\norm{\vgamma}{2} + 7 \delta^2 \sqrt{\sum_e \rho_e^4}\leq \frac{4}{300}+ 7\delta^2 \norm{\vrho}{4}^2\leq \frac{2}{150}+ \frac{7}{1089}<\frac{1}{50}.
\end{eqnarray*}
The lemma follows.
\end{proof}

\subsection{Analysis of the Algorithm}\label{sec:simple_analysis}

We want now to analyze the overall running time of our algorithm. Recall that given our target demand $\vchi$ that corresponds to sending $F$ units of flow from the source $s$ to the sink $t$, our overarching goal is to either route this demand fully in $G$ or provide a dual certificate that it is impossible to route $\vchi$ in $G$. 

We aim to achieve this goal by maintaining and gradually improving a primal dual solution $(\ff, \yy)$. In this solution, $\ff$ is a $\vchi_\alpha$-flow (which corresponds to routing an $\alpha$ fraction of the desired demand $\vchi$) that is feasible in $G$ and $\ff$ and $\yy$ are well-coupled, i.e., tied to each other via condition \eqref{eq:coupling} with the violation vector $\hgamma$ having sufficiently small $\ell_2$-norm. As described in Section \ref{sec:progress_steps}, each iteration runs in nearly-linear time and boils down to employing electrical flow computations to find an augmenting flow in the current residual graph $G_{\ff}$ (as well as to update the dual solution to maintain well-coupling).

Consequently, all we need to do now is to lower bound the amount of progress that each of these iteration makes. Ideally, we would like to prove that in each iteration in which $\ff$ already routed $\alpha$-fraction of the desired flow, i.e., $\ff$ is a feasible $\vchi_\alpha$-flow, the step size $\delta$ (see \eqref{eq:augmentation_step_update}) can be taken to be at least 
\begin{equation}\label{eq:delta_hat_def}
\delta \geq (1-\alpha) \hdelta,
\end{equation}
for some fixed $\hdelta>0$. Observe that if such a lower bound was established then, by \eqref{eq:growth_alpha}, it would imply that each iteration finds an augmenting flow that routes at least $\hdelta$-fraction of the amount of flow still to be routed. As a result, after executing at most $O(\hdelta^{-1} \log mU)$ iterations, the remaining value of flow to be routed would be at most $1$ and thus a simple flow rounding and augmenting path finding would yield the final answer (see, e.g., \cite{Madry13}), making the overall running time be at most $\tO{\hdelta^{-1}m \log U}$.

Unfortunately, a priori, it is difficult to provide any such non-trivial unconditional lower bound on the amount of primal progress we make in each iteration. After all, it could be the case that the target flow cannot be even routed in $G$. More importantly though, even if the target flow could be routed in $G$, and thus the residual graph always admitted augmenting flows of sufficiently large value, it is still not clear that our flow augmenting procedure would be able to find them. (It is worth noting that this problem is by no means specific to our algorithm. In fact, in {\em all} the maximum flow algorithms that rely on the augmenting paths framework ensuring that each iteration makes a sufficient primal progress is a chief bottleneck in the analysis.) 

The root of the problem here is that our flow augmenting procedure is based on electrical flows and these are undirected in nature. Consequently, the augmenting flows that it finds have to come from a fairly restricted class: $s$-$t$ flows that are feasible in a certain ``symmetrized" version of the residual graph. 

To make it precise, given a residual graph $G_{\ff}$, let us define its {\em symmetrization} $\hG_{\ff}$ to be an undirected graph in which each arc $e$ has its forward and backward capacity equal to $\hu_e(\ff)$, i.e., to the minimum of the forward $\hu_e^+(\ff)$ and backward $\hu_e^-(\ff)$ residual capacities in $G_{\ff}$.  Observe now that each (electrical) augmenting flow $\delta\tff$ found in the augmentation step (cf. \eqref{eq:augmentation_step_update}) is not only feasible in the residual graph $G_{\ff}$ but also in its symmetrization  $\hG_{\ff}$ -- this is exactly what the condition \eqref{eq:l_infty_bound} enforces. 

However, not all augmenting $s$-$t$ flows that are feasible in $G_{\ff}$ have to be feasible in $\hG_{\ff}$ too. In fact, it can happen that a large maximum $s$-$t$ flow value that the residual graph $G_{\ff}$ supports mostly vanishes in its symmetrization $\hG_{\ff}$, and thus prevents our algorithm from making a sufficient good primal progress. (Again, a difficulty of a exactly the same nature arises in the analysis of the classic flow augmenting algorithms such as \cite{EvenT75,Karzanov73,GoldbergR98}.)


\paragraph{Preconditioning Arcs.} It turns out, however, that there is a fairly simple way to circumvent the above difficulty and ensure that the kind of direct, ``primal-only`` analysis we hoped for above can indeed be performed. Namely, we just need to ``precondition`` our input graph by adding to it a large number of $s$-$t$ arcs of sufficiently large capacities. 

More precisely, we modify our input graph $G$ by adding to it $m$ undirected arcs between the source $s$ and sink $t$ with a forward and backward capacities equal to $2U$ and their orientation being from $s$ to $t$. We will call these arcs {\em preconditioning arcs}. Observe that after this modification the number of arcs of our graph as well as its maximum capacity at most doubled, and the maximum $s$-$t$ flow value changed additively by exactly $2mU$. In particular, the preconditioning arcs constitute exactly half of all the arcs and the amount of $s$-$t$ flow that they alone can support is at least twice the $s$-$t$ throughput of the rest of the graph. (Also, as these arcs are undirected they do not interfere with our initialization procedure -- cf. Lemma \ref{lem:initialization}.\footnote{More precisely, we can just first initialize our framework for the original graph, as before, and only precondition the undirected graph resulting from that initialization.}) Consequently, we can just focus on analyzing the running time of our algorithm on this preconditioned instance and the bounds we establish will immediately translate over to the original instance.\footnote{Note that the preconditioning arcs have to be fully saturated in any maximum $s$-$t$ flow. So, simply dropping these arcs and the flow on them will yield the maximum $s$-$t$ flow in the original graph.}

As already mentioned, somewhat surprisingly, once such preconditioning arcs are in place and our primal dual solution $(\ff, \yy)$ is well-coupled, it is always the case that the symmetrization $\hG_{\ff}$ of our current residual graph $G_{\ff}$ retains a constant fraction of the $s$-$t$ throughput. Intuitively speaking, well-coupling prevents the ``shortcutting'' preconditioning arc from getting ``clogged'' too quickly. Instead, their residual capacity is consumed at the same rate as that of the rest of the graph. Consequently, these arcs alone are always able to provide enough of $s$-$t$ throughput in the symmetrization $\hG_{\ff}$ of the residual graph $G_{\ff}$. This is made precise in the following lemma, whose proof appears in Appendix \ref{app:symmetrization_throughput}



\begin{lemma}\label{lem:symmetrization_throughput}
Let $(\ff, \yy)$ be a well-coupled primal dual solution in the (preconditioned) graph $G$ and let $\ff$ be a $\vchi_\alpha$-flow, for some $0\leq \alpha<1$, that is feasible in $G$. We have either that:
\begin{enumerate}[(a)]
\item there exists a $\vchi_{\frac{(1-\alpha)}{10}}$-flow $\ff'$ that is feasible in the symmetrization $\hG_{\ff}$ of the residual graph $G_{\ff}$;
\item or $\vchi^T\yy>\frac{2m}{(1-\alpha)}$ implying that our target demand $\vchi$ cannot be routed in $G$ (cf. Lemma \ref{lem:duality_certificate}).
\end{enumerate}
\end{lemma}

Note that if our target demand $\vchi$ is exactly the demand $F^*\vchi_{s,t}$ of the maximum $s$-$t$ flow, the second condition cannot ever trigger and thus indeed it is the case that the symmetrization of the (preconditioned) residual graph retains a constant fraction of the original $s$-$t$ throughput.

\paragraph{Lower Bounding $\hdelta$.} Once we proved that the symmetrization $\hG_{\ff}$ of the residual graph $G_{\ff}$ retains most of its $s$-$t$ flow throughput (see Lemma \ref{lem:symmetrization_throughput}), we are finally able to provide an absolute lower bound $\hdelta$ (cf. \eqref{eq:delta_hat_def}) on the amount of primal progress each iteration of our algorithm makes. To this end, we upper bound first the energy, or, (almost) equivalently, the $\ell_2$-norm of the congestion vector (see Lemma \ref{lem:energy_rho}) of the electrical flow that we use in our augmentation step (see \eqref{eq:augmentation_step_update}). 

\begin{lemma}
\label{lem:energy_upperbound}
Let $(\ff, \yy)$ be a well-coupled primal dual solution, with $\ff$ being a $\vchi_\alpha$-flow that is feasible in $G_{\ff}$, for some $0\leq \alpha <1$. Let $\tff$ be an electrical $\vchi$-flow determined by the resistances $\rr$ given by \eqref{eq:res_definition}. We have that either:
\begin{enumerate}[(a)]
\item $\norm{\vrho}{2}^2\leq \energy{\rr}{\tff} \leq \frac{\cenergy m}{(1-\alpha)^2}$, where $\vrho$ is the congestion vector defined in \eqref{eq:congestion_vec_def}, and $\cenergy>0$ is an explicit constant;
\item or, $\vchi^T\yy>\frac{2m}{(1-\alpha)}$, i.e., our target demand $\vchi$ cannot be routed in $G$.
\end{enumerate}
\end{lemma}

\begin{proof}
By Lemma \ref{lem:symmetrization_throughput}, we either have that $\vchi^T\yy>\frac{2m}{(1-\alpha)}$, i.e., our condition (b) triggers, or there exists a $\vchi_{\frac{(1-\alpha)}{10}}$-flow $\ff'$ that is feasible in the symmetrization $\hG_{\ff}$ of the residual graph $G_{\ff}$. 

Let us first bound the energy $\energy{\rr}{\ff'}$ of that flow $\ff'$. To this end, observe that, by definition \eqref{eq:res_definition}, we have that
\[
\energy{\rr}{\ff'} = \sum_e r_e (f_e')^2 \leq \sum_e \left(\frac{1}{(\hu_e^+(\ff))^2}+\frac{1}{(\hu_e^-(\ff))^2}\right) (f_e')^2 \leq 2 \sum_e \left(\frac{f_e'}{\hu_e(\ff)}\right)^2 \leq 2m,
\]
where the last inequality follows by the fact that the flow $\ff_e'$ is feasible in $\hG_{\ff}$. (Recall that the forward and backward capacities of an arc $e$ in $\hG_{\ff}$ are equal to $\hu_e(\ff)$.)

Consequently, if we scale $\ff'$ by $\frac{10}{(1-\alpha)}$ then it will become a $\vchi$-flow and its energy will be at most $\frac{200}{(1-\alpha)^2}m$. However, by definition, the energy $\energy{\rr}{\tff}$ of the electrical $\vchi$-flow $\tff$ cannot be larger than the energy of any other $\vchi$-flow. We thus can conclude that
\[
\norm{\vrho}{2}^2 \leq \energy{\rr}{\tff} \leq \frac{100}{(1-\alpha)^2}\energy{\rr}{\ff'}\leq \frac{200}{(1-\alpha)^2}m,
\]
where we also used Lemma \ref{lem:energy_rho}. So, taking $\cenergy:=200$ concludes the proof.
\end{proof}

Now, we should notice that by Lemma \ref{lem:l_4_condition_step} it suffices that we always have that 
\begin{equation}\label{eq:l_2_runntime_bound}
\delta \leq \frac{1}{33\norm{\vrho}{4}} \leq \frac{1}{33\norm{\vrho}{2}}\leq \frac{(1-\alpha)}{33 \sqrt{\cenergy m}},
\end{equation}
where we used Lemma \ref{lem:energy_upperbound} and the simple fact that it is always the case that $\norm{\vrho}{4}\leq \norm{\vrho}{2}$. Consequently, by \eqref{eq:delta_hat_def}, we see that we can take $\hdelta:=(33 \sqrt{\cenergy m})^{-1}$, which gives us the desired $\tO{m^{\frac{3}{2}} \log U}$ time algorithm. 

Finally, we want to emphasize again that even though our above analysis was based solely on analyzing our primal progress\footnote{In fact, one could perform it even without resorting explicitly to the dual infeasibility certificates $\vchi^T \yy$.}, maintaining the dual solution and the primal dual coupling \eqref{eq:coupling} was absolutely critical for its success.

\section{An Improved  \texorpdfstring{$\tO{m^{\frac{10}{7}}U^{\frac{1}{7}}}$}{\~O(m\textasciicircum(10/7) U\textasciicircum(1/7)}-Time  Algorithm}\label{sec:improved_algorithm}

In this section, we present an improved maximum $s$-$t$ flow algorithm that runs in time $\tO{m^{\frac{10}{7}}U^{\frac{1}{7}}}$, where $U$ is the value of the largest (integer) capacity of the graph. In the case of the unit capacities, i.e., $U=1$, this running time matches the running time of the $\tO{m}^{\frac{10}{7}}$-time algorithm of M\k{a}dry \cite{Madry13}, and improves over the best known bound of $\min\{\tO{(mU)^{\frac{10}{7}}}, \tO{m\sqrt{n}\log U}\}$ that stems for the work of M\k{a}dry \cite{Madry13} and Lee and Sidford \cite{LeeS14} whenever $U$ is moderately large (and the graph is sufficiently sparse).

To achieve this goal, we will modify the $\tO{m^{\frac{3}{2}}\log U}$-time algorithm presented in Section \ref{sec:basic_algorithm}. These modifications  will enable us to lower bound the amount of primal progress $\hdelta$ (cf. \eqref{eq:delta_hat_def}) this algorithm makes in each iteration. Specifically, we aim to show that
\begin{equation}
\label{eq:improved_del_hat_bound}
\hdelta \geq \left(m^{\frac{1}{2}-\eta}\right)^{-1},
\end{equation}
for certain $\eta:=\frac{1}{14}-\frac{1}{7} \log_m U - O(\log \log mU)$, whenever that modified algorithm executes progress steps (see Section \ref{sec:progress_steps}). Clearly, as discussed in Section \ref{sec:simple_analysis}, establish that lower bound on $\hdelta$ ensures that we need at most $\tO{m^{\frac{1}{2}-\eta}\log U}=\tO{m^{\frac{10}{7}}U^{\frac{1}{7}}}$ progress steps to converge to the optimum solution. 

\subsection{Modified Algorithm}\label{sec:modified_algorithm}

Before we explain how we modify the $\tO{m^{\frac{3}{2}}\log U}$-time algorithm we presented in Section \ref{sec:basic_algorithm}, let us understand first what is its main running time bottleneck. 

To this end, recall that the key quantity that captures the progress that each augmentation iteration makes is the $\ell_4$-norm $\norm{\vrho}{4}$ of the congestion vector $\vrho$ of the electrical $\vchi$-flow $\tff$ determined by the resistances $\rr$ (see \eqref{eq:res_definition}). More precisely, by Lemma \ref{lem:l_4_condition_step},  we need to always have that
\[
(1-\alpha) \hdelta\leq \delta \leq (33\norm{\vrho}{4})^{-1}.
\]

In Section \ref{sec:simple_analysis}, we bounded that $\ell_4$-norm by simply noticing that it always has to be upper bounded by the $\ell_2$-norm of that vector and that, by Lemma \ref{lem:energy_rho}, this $\ell_2$-norm is directly tied to the energy $\energy{\rr}{\tff}$ of the electrical flow $\tff$. Furthermore, since we were able to also prove that this energy is always $O\left(\frac{m}{(1-\alpha)^2}\right)$ (see Lemma \ref{lem:energy_upperbound}) the final $O(\sqrt{m})$ bound on $\hdelta^{-1}$ followed. (See \eqref{eq:l_2_runntime_bound} for details.) 

Now, the key observation is that such bounding of the $\ell_4$-norm of the vector $\vrho$ with its $\ell_2$-norm might be wasteful. That is, even though it is not hard to construct examples in which the congestion vector $\vrho$ has these two norms be fairly close, such examples are inherently fragile. More precisely, as first pointed out by Christano et al. \cite{ChristianoKMST11} in the context of an analogous $\ell_\infty$- vs $\ell_2$-norm trade off, whenever $\norm{\vrho}{4}\approx \norm{\vrho}{2}$, it must be necessarily be the case that most of the energy of the electrical flow $\tff$ is contributed by a very small set of arcs. Moreover, if one perturbs such arcs by increasing their resistance, this will result in a great increase of the energy of the corresponding electrical flow. Christano et al. \cite{ChristianoKMST11}, and later M\k{a}dry \cite{Madry13} have demonstrated that a careful exploitation of this phenomena can lead to ensuring that such fragile worst-case instances do not appear too often and thus a better bound on the $\ell_4$-norm can be achieved.

\paragraph{Arc Boosting.} Our improved algorithm will also follow the perturbation paradigm we outlined above. To this end, we introduce an operation of boosting that we will use to perturb high-energy contributing arcs whenever needed. 

Formally, let us fix a primal dual well-coupled solution $(\ff, \yy)$ and a particular arc $e=(u,v)$. Let us assume wlog that $\hu_e^+(\ff)=\hu_e(\ff)$ (the case of $\hu_e^-(\ff)=\hu_e(\ff)$ is completely symmetric). We define a {\em boost} of $e$ to be an operation that modifies the input graph and the solution $(\ff, \yy)$ in the following way:

\begin{enumerate}[(1)]
\item The arc $e$ is replaced by an $u$-$v$ path consisting of $\beta(e):=2+\ceil{\frac{2U}{\hu_e(\ff)}}$ arcs $e_1, \ldots, e_{\beta(e)}$, all oriented towards $v$.
\item The first two arcs $e_1, e_2$ are just copies of the arc $e$, with $u_{e_1}^+=u_{e_2}^+:=u_{e}^+$ and $u_{e_1}^-=u_{e_2}^-:=u_{e}^-$. 
\item For all the remaining $(\beta(e)-2)$ arcs, we have that $u_{e_i}^+:=+\infty$ and $u_{e_i}^-:=\tu$, for $2< i \leq \beta(e)$, where $\tu\geq U$ is a value that we will set below.
\item The flow $\ff$ pushes the original flow $f_e$ over the created $u$-$v$ path. 
\item Finally, if $v_0=u, v_1, \ldots, v_{\beta(e)}=v$ are the consecutive vertices on the created path then: 
\begin{itemize}
\item $y_{v_0}:=y_u$ and $y_{v_{\beta(e)}}:=y_v$;
\item $y_{v_1}:=y_v$, and $y_{v_2}:=y_{v_1}+\Phi_{e}(\ff)=y_{v}+\Phi_{e}(\ff)$;
\item $y_{v_3}, \ldots, y_{v_{\beta(e)-1}}$ are set so as $\Delta_{e_i}(\yy)=-\frac{\Delta_{e_2}(\yy)}{\beta(e)-2}=-\frac{\Phi_{e}(\ff)}{\beta(e)-2}$, for each $2<i\leq \beta(e)$. 
\end{itemize}
\end{enumerate}

Observe that after the above modifications, we have that
\begin{eqnarray*}
\Delta_{e_1}(\yy)-\Phi_{e_1}(\ff)& = &\Delta_{e}(\yy)-\Phi_{e}(\ff)\\
\Delta_{e_2}(\yy)&=&\Phi_{e_2}(\ff)=\Phi_{e}(\ff)\\
\Phi_{e_i}(\ff) &= & -\frac{1}{(\tu+f_e)}=\Delta_{e_i}(\yy), 
\end{eqnarray*}  
 for each $2<i\leq \beta(e)$, as long as we set 
 \[
 \tu:=\frac{\beta(e)-2}{\Phi_{e}(\ff)}-f_e.
 \] 
Therefore, the solution $(\ff, \yy)$ remains feasible in the (modified) residual graph $G_{\ff}$ and is well-coupled too. Also, note that since $\Phi_e(\ff)\leq \frac{1}{\hu_e(\ff)}$ and $|f_e|\leq U$, the fact that $(\beta(e)-2)\geq \frac{2U}{\hu_e(\ff)}$ ensures that $u_{e_i}^-=\tu\geq U$, for all $2<i\leq \beta(e)$. Consequently, the effective forward and backward capacities of the created path are the same as the forward and backward capacities of the original arc $e$, and thus the $s$-$t$ throughput of our graph remains unchanged. 

We thus see that, from our point of view, the boosting operation had only two real consequences: one of them is the intended effect and the other one is an undesired side effect. The intended effect is that the effective resistance with respect to the resistances $\rr$ (cf. \eqref{eq:res_definition}) of the created path is at least twice as large as the original resistance $r_e$ of the arc $e$. (Note that the arcs $e_1$ and $e_2$ have exactly the same residual capacities, and thus the resistances $r_{e_1}$ and $r_{e_2}$ are both equal to the resistance $r_e$ of the original arc $e$.) The undesired consequence, however, is that the number of arcs in our graph increased by $\beta(e)-1$.

The key impact of that (at least) doubling of the resistance  of the arc $e$ is that, intuitively, it makes the energy of our electrical $\vchi$-flow $\tff$ increase proportionally to the initial contribution of that arc to the (squared) $\ell_2$-norm of the congestion vector $\vrho$. This is made precise in the following lemma. (This lemma is stated in a slightly more general form for future reference. For now, one can think of the set $S$ consisting of a single arc $e$.)

\begin{lemma}
\label{lem:energy_increasy} Let $\tff$ be an electrical $\vchi$-flow determined by the resistance $\rr$ in our graph $G$ and let $\vrho$ be the corresponding congestion vector defined by \eqref{eq:congestion_vec_def}. If $\tff'$ is the electrical $\vchi$-flow determined by the resistances $\rr'$ in the graph $G'$ that was obtained by boosting a set of arcs $S$ then
\[
\energy{\rr'}{\tff'}\geq \left(1+\frac{1}{8}\left(\sum_{e\in S}\frac{\rho_e^2}{\norm{\vrho}{2}^2}\right) \right)\energy{\rr}{\tff}.
\] 
\end{lemma}

\begin{proof}
Let $\tvphi$ be the vertex potentials inducing $\tff$, by Lemma \ref{lem:effective_conductance} and \eqref{eq:congestion_vec_def}, we have that
\[
\frac{1}{\energy{\rr}{\tff}} = \sum_{e'=(u',v')} \frac{(\tphi_{v'}-\tphi_{u'})^2}{r_{e'} \energy{\rr}{\tff}^2}
\]
and $\vchi^T \tvphi = \energy{\rr}{\tff}$. 

We want now to use the same voltages $\tvphi$ to lowerbound the energy $\energy{\rr'}{\tff'}$ of the electrical $\vchi$-flow $\tff'$ after boosting the arcs in $S$. Observe that the effective resistance $r_{u,v}'$ between the endpoints $u$ and $v$ of some boosted arc $e=(u,v) \in S$ after that boost is at least $2r_e$. So, when we analyze the energy of $\tff'$ we can treat the $u$-$v$ path that was created as just an arc $e$ but with the resistance $r_e'$ equal to $r_{u,v}'\geq 2r_e$. 

Therefore, by using Lemma \ref{lem:effective_conductance} again, we can conclude that
\begin{eqnarray*}
\frac{1}{\energy{\rr'}{\tff'}} &\leq& \left(\sum_{e'=(u',v')\notin S} \frac{(\tphi_{v'}-\tphi_{u'})^2}{r_{e'} \energy{\rr}{\tff}^2}\right)+ \sum_{e=(u,v)\in S}\frac{(\tphi_v-\tphi_u)^2}{r_e' \energy{\rr}{\tff}^2}\\
&\leq& \left(\sum_{e'=(u',v')} \frac{(\tphi_{v'}-\tphi_{u'})^2}{r_{e'} \energy{\rr}{\tff}^2}\right)- \sum_{e=(u,v)\in S}\frac{(\tphi_v-\tphi_u)^2}{2r_e \energy{\rr}{\tff}^2} =\frac{1}{\energy{\rr}{\tff}} \left(1 - \sum_{e=(u,v)\in S}\frac{(\tphi_v-\tphi_u)^2}{2r_e \energy{\rr}{\tff}}\right)\\
&\leq& \frac{1}{\energy{\rr}{\tff}} \left(1 - \sum_{e\in S}\frac{r_e \tf_e^2}{2 \energy{\rr}{\tff}}\right)\leq \frac{1}{\energy{\rr}{\tff}} \left(1 - \sum_{e\in S}\frac{\rho_e^2}{4 \norm{\vrho}{2}^2}\right),
\end{eqnarray*}
where we also used Lemma \ref{lem:energy_rho} and Ohm's law \eqref{eq:Ohms_lawy}. The lemma follows by noticing that $\frac{1}{(1-\eps)}\geq (1+\frac{\eps}{2})$ whenever $0\leq \eps\leq \frac{1}{4}$.
\end{proof}

\paragraph{Boosting Use and Early Termination.} Once we introduced the boosting operation, we can state the two modifications of the basic algorithm from Section \ref{sec:basic_algorithm} that we will make.

The first one is relatively minor: we terminate our algorithm whenever 
\begin{equation}
\label{eq:early_termination}
(1-\alpha)F \leq m^{\frac{1}{2}-\eta}.
\end{equation} 
That is, once we know that we are within an additive factor of $m^{\frac{1}{2}-\eta}$ of the target flow, we do not execute progress steps anymore. Observe that in that case we can just round our current flow and compute the optimal solution using at most $m^{\frac{1}{2}-\eta}$ augmenting paths computation (see, e.g., \cite{Madry13} for details). So, if we perform these operations, our algorithm will still run within the desired running time bound.

The second modification specifies how and when we apply boosting operations. To this end, given an electrical $\vchi$-flow $\tff$ and its congestion vector $\vrho$, let us define an arc $e$ to be {\em high-energy} if 
\begin{equation}
\label{eq:def_high_energy}
|\rho_e|\geq \rho^* = \frac{m^{\frac{1}{2}-3\eta}}{\crho(1-\alpha)},
\end{equation}
for certain constant $\crho>0$ that we will set later (see \eqref{eq:non_high_energy_contribution} below). 

Now, recall that  our goal is to ensure that \eqref{eq:improved_del_hat_bound} holds, which corresponds to ensuring that, by Lemma \ref{lem:l_4_condition_step}, $\norm{\vrho}{4}\leq \frac{m^{\frac{1}{2}-\eta}}{33(1-\alpha)}$. However, for reasons that will be clear shortly, we will actually want to have a stronger, $\ell_3$-norm bound on $\vrho$. Specifically, we want to make a progress step only when
\begin{equation}
\label{eq:l_3_step_size}
\norm{\vrho}{3} \leq \frac{m^{\frac{1}{2}-\eta}}{33(1-\alpha)},
\end{equation}
and then use the step size 
\[
\delta=\frac{\hdelta}{(1-\alpha)}= \frac{1}{33 (1-\alpha) \norm{\vrho}{3}}
\]
when making the augmentation step described in Section \ref{sec:progress_steps}. Observe that as $\norm{\vrho}{3}\geq \norm{\vrho}{4}$, this new condition still satisfies the $\ell_4$-norm requirement set forth in Lemma \ref{lem:l_4_condition_step}.

Unfortunately, as mentioned earlier, it can sometime happen that the bound \eqref{eq:l_3_step_size} is violated. Therefore, we want to use our boosting technique to enforce that such violations do not happen too often. Specifically, whenever condition \eqref{eq:l_3_step_size} does not hold we boost all the arcs in the set $S^*$, which is defined to be the set of up to $k^*:=m^{4\eta}$ high-energy arcs $e$ with highest values of $|\rho_e|$. In other words, we order all the high-energy arcs in a non-increasing order with respect to $|\rho_e|$ and take $S^*$ to consist of the first up to $k^*$ of them. (Note that, in principle, in some iterations there can be only very few, or even no, high-energy arcs.) 

Once we boost the set $S^*$, we simply proceed to the next iteration (in the new version of the graph that our boosting operation created). Note that, at first glance, there is nothing preventing us from needing to execute such a boosting step in each iteration, and thus to never be able to execute progress steps. However, as we argue in the next section, we are actually guaranteed to be able to execute progress steps often enough and thus make sufficiently good progress overall.

\subsection{Analysis of the Improved Algorithm}\label{sec:improved_analysis}

We proceed now to analyzing our modified algorithm and establishing the improved running time bound. This will require tackling two issues. 

First one is controlling the increase in the number of arcs resulting from performing of all our boosting operations. Specifically, we want to maintain the following invariant.

\begin{invariant}
\label{inv:arc_number}
The total increase in the number of arcs is at most $\frac{m}{10}$. 
\end{invariant}

As we will see, ensuring that this invariant is maintained will be what determines our final bound on the value of $\eta$ (cf. \eqref{eq:improved_del_hat_bound}).

We need to maintain this invariant in order to be able to use the machinery we developed in Section \ref{sec:basic_algorithm}. (It is not hard to check that, as long as that invariant is not violated, the whole analysis performed there is still valid, subject to slight adjustment of the corresponding constants.)

The other issue we need to tackle corresponds to bounding the total number of boosting steps we execute. After all,  we want to ensure that we execute progress steps often enough. In order to obtain such a bound we will perform a potential based argument with our potential function being the energy $\energy{\rr}{\tff}$ of the electrical $\vchi$-flow $\tff$ we compute in each iteration. (Note that this electrical flow $\tff$ is always determined by the resistances $\rr$ (cf. \eqref{eq:res_definition}) in the current version of the graph $G$.) One can easily convince oneself  that this energy has to be at least $\frac{1}{U^2}$ at the beginning of the algorithm and, by Lemma \ref{lem:energy_upperbound} (and provided Invariant \ref{inv:arc_number} holds), it is never larger than 
\[
\energy{\rr}{\tff}\leq \frac{\cenergy m}{(1-\alpha)^2}\leq O(m^3 U^2).
\]
(Here,  the second inequality follows  due to the fact that $1\leq F\leq 2mU$ and we have our termination condition \eqref{eq:early_termination}.) We want now to establish two claims:

\begin{enumerate}[(I)]
\item Whenever a boosting step is executed, the energy $\energy{\rr}{\tff}$ increases by a factor of at least $\left(1+\cincrease m^{-2\eta}\right)$, for some constant $\cincrease>0$. (To prove this statement we will use Lemma \ref{lem:energy_increasy}.)
\item Whenever a progress step is performed, with the condition \eqref{eq:l_3_step_size} satisfied, the energy $\energy{\rr}{\tff}$ decreases by at most $\left(1+\cdrop m^{-2\eta}\right)$, for some other constant $\cdrop>1$. (Proving this statement is where we leverage the fact that we enforce the stronger, $\ell_3$-norm condition \eqref{eq:l_3_step_size} on $\vrho$.)
\end{enumerate}

Observe that once the above two claims are established, we can amortize the number of boosting steps against the number of good progress steps. Specifically, a simple calculation shows that unless the number of boosting step is within a $O(\log m \log U)$ factor of the number of the progress steps, the resulting increase of the energy $\energy{\rr}{\tff}$ would violate the absolute upper bound $O(m^3 U^2)$ we established. Furthermore, since we execute a progress step only when condition \eqref{eq:l_3_step_size} is satisfied, this means that our desired lower bound \eqref{eq:improved_del_hat_bound} on $\hdelta$ holds. Therefore, by our reasoning in Section \ref{sec:simple_analysis}, we know that the number of progress steps is at most $\tO{\hdelta^{-1}\log mU}=\tO{m^{\frac{1}{2}-\eta}\log mU}$. As each boosting and progress step can be executed in nearly-linear time, our desired improved running time bound will then follow. 

In the light of the above, it remains to only argue that claims (I) and (II) are indeed correct and that the Invariant \ref{inv:arc_number} is never violated.

\paragraph{Preservation of the Invariant \ref{inv:arc_number}.} To prove that the Invariant \ref{inv:arc_number} is indeed preserved, let us recall that each arc boosting operation increases the number of arcs by $\beta(e)-1\leq \frac{2U}{\hu_e(\ff)}$, where $\hu_e(\ff)$ is the residual capacity of the arc $e$ at the time the boosting occurred. As it turns out, since we only boost arcs that are high-energy, i.e., arcs $e$ with $|\rho_e|\geq \rho^*$ at the time of the boosting, one can use a simple vertex potential based argument to argue that $\beta(e)$ has to be always at most $O(m^{4\eta}U)$. 

\begin{lemma}
\label{lem:measure_bound}
Assume that Invariant \ref{inv:arc_number} holds, let $\tff$ be an electrical $\vchi$-flow determined by the resistances $\rr$ and let $e=(u,v)$ be a high-energy arc that was boosted. Then, the total arc number increase  $\beta(e)-1$ is at most $O(m^{4\eta}U)$.
\end{lemma}

\begin{proof}
Recall that $\beta(e)-1\leq \frac{2U}{\hu_e(\ff)}$. So, it suffices that we prove that $\hu_e(\ff)\geq \Omega( m^{-4\eta})$. To this end, let $\tvphi$ be vertex potentials inducing the electrical flow $\tff$ and let $R_{s,t}$ be the effective resistance between the source $s$ and sink $t$ in our graph. As $\tff$ is an $s$-$t$ flow of value $F$, we need to have that
\[
F^2 R_{s,t}=\energy{\rr}{\tff} \leq \frac{\cenergy m}{(1-\alpha)^2},
\]
where we used Lemma \ref{lem:energy_upperbound} and the fact that Invariant \ref{inv:arc_number} holds. 

On the other hand, we know that the vertex potential drop $\tphi_t-\tphi_s$ between $s$ and $t$ is non-negative (as the electrical flow has to flow from $s$ to $t$) and at most 
\[
\tphi_t-\tphi_s = F R_{s,t} \leq \frac{\cenergy m}{(1-\alpha)^2F} \leq \frac{\cenergy m^{\frac{1}{2}+\eta}}{(1-\alpha)},
\]
where the last inequality follows as $(1-\alpha)F\geq m^{\frac{1}{2}-\eta}$ due to our termination condition \eqref{eq:early_termination}. 

However, as $\tff$ is an $s$-$t$ flow it must be the case that the vertex potential difference between the endpoints of our arc $e=(u,v)$ cannot be larger than such difference between the vertex potentials of $s$ and $t$. That is, we have that 
\begin{equation}\label{eq:potential_drop_bound}
|\tphi_v-\tphi_u| \leq |\tphi_t-\tphi_s|\leq \frac{\cenergy m^{\frac{1}{2}+\eta}}{(1-\alpha)}.
\end{equation}
But, by Ohm's law \eqref{eq:Ohms_lawy} and definition \eqref{eq:res_definition} of the resistance $\rr$, we have that
\[
|\tphi_v-\tphi_u| = |\tf_e| r_e \geq \frac{|\tf_e|}{\hu_e(\ff)^2} = \frac{|\rho_e|}{\hu_{e}(\ff)} \geq \frac{\rho^*}{\hu_e(\ff)}\geq  \frac{m^{\frac{1}{2}-3\eta}}{\crho(1-\alpha)\hu_e(\ff)},
\]
where we used the definition \eqref{eq:def_high_energy} of the high-energy arcs.

So, putting this inequality together with the bound \eqref{eq:potential_drop_bound}, gives us the desired lower bound on $\hu_e(\ff)$. The lemma follows.
\end{proof}

Observe now that as our boosting operation never boosts more than $k^*=m^{4\eta}$ arcs at a time and we have at most $\tO{m^{\frac{1}{2}-\eta}\log U}$ boosting steps, by Lemma \ref{lem:measure_bound}, the total arc number increase throughout the whole algorithm is at most
\[
m^{4\eta} \cdot \tO{m^{\frac{1}{2}-\eta}\log U} \cdot O(m^{4\eta} U) = \tO{m^{\frac{1}{2}+7\eta}U}<\frac{m}{10},
\]
provide we set $\eta:=\frac{1}{14}-\frac{1}{7} \log_m U - O(\log \log mU)$ with appropriately chosen constant in the $O(\log \log mU)$ term. Thus Invariant \ref{inv:arc_number} is indeed never violated.

\paragraph{Establishing Claim (I).} We want now to prove that each boosting step results in an increase of the energy $\energy{\rr}{\tff}$ by a factor of at least $\left(1+\Omega(m^{-2\eta})\right)$. In the light of Lemma \ref{lem:energy_increasy}, it suffices we show that if $S^*$ is the set of high-energy arcs that got boosted due to the condition \eqref{eq:l_3_step_size} not holding then these arcs had to have a large contribution to the (square) of the $\ell_2$-norm of the congestion vector $\vrho$. That is, we want to argue that
\begin{equation}
\label{eq:S_star_contribution}
\sum_{e\in S^*} \frac{\rho_e^2}{\norm{\vrho}{2}^2} \geq \Omega(m^{-2\eta}).
\end{equation}
Clearly, once we establish that then the desired energy increase will follow from Lemma \ref{lem:energy_increasy} applied to $S=S^*$.

To prove \eqref{eq:S_star_contribution}, let us first consider the case that $S^*$ has maximum size, i.e., $|S^*|=k^*=m^{4\eta}$. Note that all the arcs in $S^*$ are high-energy arcs. Therefore, \eqref{eq:def_high_energy} gives us in this case that 
\[
\sum_{e\in S^*} \frac{\rho_e^2}{\norm{\vrho}{2}^2} \geq m^{4\eta}\cdot  \frac{(\rho^*)^2}{\norm{\vrho}{2}^2} \geq \frac{m^{1-2\eta}}{\crho^2(1-\alpha)^2\norm{\vrho}{2}^2}\geq \Omega(m^{-2\eta}),
\]
where we used the fact that Lemma \ref{lem:energy_upperbound} implies that $\norm{\vrho}{2}^2 \leq \frac{\cenergy m}{(1-\alpha)^2}$. So, \eqref{eq:S_star_contribution} holds in this case and we can focus on the situation when $|S^*|<m^{4\eta}$. Observe that this means that there is no high-energy arcs that are not in $S^*$.

We want to argue that in such situation the arcs of $S^*$ contribute at least a half of the total (cubed) $\ell_3$-norm of the vector $\vrho$. That is, that we have that
\begin{equation}\label{eq:S_star_l_3_norm_bound}
\sum_{e\in S^*} |\rho_e|^3 \geq \left(\frac{m^{\frac{1}{2}-\eta}}{66(1-\alpha)}\right)^3.
\end{equation}

To this end, observe that the contribution of arcs that are not high energy (and thus not in $S^*$) to that (cubed) $\ell_3$-norm is, by the Cauchy-Schwartz inequality, \eqref{eq:def_high_energy} and Lemma \ref{lem:energy_upperbound}, at most
\begin{equation}\label{eq:non_high_energy_contribution}
\left(\sum_{e\notin S^*} |\rho_e|^3\right)\leq \left(\max_{e\notin S^*} |\rho_e|\right) \left(\sum_{e\notin S^*} \rho_e^2\right) \leq \rho^* \norm{\vrho}{2}^2 \leq \frac{\cenergy m^{\frac{3}{2}-3\eta}}{\crho(1-\alpha)^3} \leq \left(\frac{m^{\frac{1}{2}-\eta}}{66(1-\alpha)}\right)^3,
\end{equation}
provided the constant $\crho$ is set to be large enough. Therefore, since the condition \eqref{eq:l_3_step_size} does not hold, we must have that the inequality \eqref{eq:S_star_l_3_norm_bound} is indeed valid.

Now, if there exists an arc $e\in S^*$ such that $|\rho_e|\geq \frac{m^{\frac{1}{2}-\eta}}{(1-\alpha)}$ then the contribution of this arc alone suffices to make \eqref{eq:S_star_contribution} true. We can thus assume that $\max_{e\in S^*}|\rho_e| < \frac{m^{\frac{1}{2}-\eta}}{(1-\alpha)}$ and apply the Cauchy-Schwartz inequality to \eqref{eq:S_star_l_3_norm_bound} to obtain that
\[
\sum_{e\in S^*} \rho_e^2 \geq \frac{\sum_{e\in S^*} |\rho_e|^3}{\max_{e\in S^*} |\rho_e|} \geq \frac{1}{66}\left(\frac{m^{\frac{1}{2}-\eta}}{66(1-\alpha)}\right)^2,
\]
which, by Lemma \ref{lem:energy_upperbound}, establishes \eqref{eq:S_star_contribution}. Claim (I) is thus proved.

\paragraph{Establishing Claim (II).} To prove Claim (II) we will tie the potential decrease of energy during progress steps to the quantity that we already control: the $\ell_3$-norm of the congestion vector $\vrho$. (In fact, the desire to control the energy decrease is precisely the reason why we enforce the stronger $\ell_3$-norm condition \eqref{eq:l_3_step_size} instead of the $\ell_4$-norm condition that Lemma \ref{lem:l_4_condition_step} suggests.)

We make this connection precise in the following lemma, whose proof appears in Appendix \ref{app:energy_decrease_control}.

\begin{lemma}
\label{lem:energy_decrease_control}
Let $(\ff, \yy)$ be a well-coupled feasible primal dual solution and let $\tff$ be the corresponding electrical $\vchi$-flow determined by the resistances $\rr$ and congestion vector $\vrho$ that satisfies the $\ell_3$-norm condition \eqref{eq:l_3_step_size}. Then, after execution of the progress steps, as described in Section \ref{sec:progress_steps}, the energy $\energy{\rr}{\tff}$ decreases by a factor of at most
\[
\left(1+O\left(\frac{\norm{\vrho}{3}^2}{\norm{\vrho}{2}^2}\right)\right).
\]
\end{lemma}

Observe that, by \eqref{eq:l_3_step_size} and Lemma \ref{lem:energy_upperbound}, we have that
\[
\frac{\norm{\vrho}{3}^2}{\norm{\vrho}{2}^2}\leq \frac{m^{1-2\eta}(1-\alpha)^2}{(33)^2\cenergy m (1-\alpha)^2}=O(m^{-2\eta}),
\]
which immediately establishes Claim (II) and thus concludes the analysis of the improved algorithm.

\vspace{10pt}

\noindent{\bf Acknowledgments.} We are grateful to Michael Cohen, Slobodan Mitrovi\'c, Dimitris Tsipras, and Adrian Vladu for a number of helpful discussions on this topic.

\appendix
\newpage
\section{Appendix}\label{app:main}

\subsection{Proof of Fact \ref{fa:Taylor_approximation}}\label{app:Taylor_approximation}

Let us define
\[
g(x):=\left(\frac{1}{u_1-x}-\frac{1}{u_2+x}\right).
\]
By Taylor's theorem, we have that
\begin{equation}\label{eq:g_Taylor_approx}
g(x) = g(0)+g'(0) x+g''(z)\frac{x^2}{2} =\frac{1}{u_1}-\frac{1}{u_2}+\left(\frac{1}{u_1^2}+\frac{1}{u_2^2}\right)x+x^2\left(\frac{1}{(u_1-z)^3}-\frac{1}{(u_2+z)^3}\right)
\end{equation}
where $z$ is some value with $|z|\leq |x|$. As a result, we can conclude that
\begin{equation}\label{eq:g_tail_Taylor_approx}
|\zeta|=\left|\left(\frac{1}{(u_1-z)^3}-\frac{1}{(u_2+z)^3}\right)\right| \leq \left(\frac{64}{27 u_1^3}+\frac{64}{27 u_2^3}\right)\leq \frac{5}{u^3},
\end{equation}
as we wanted to show.

\subsection{Proof of Lemma \ref{lem:augmentation_coupling}}\label{app:augmentation_coupling}

By the definition of our update \eqref{eq:augmentation_step_update} and the approximation bound \eqref{eq:Taylor_approx} as well as the coupling condition \eqref{eq:coupling}, we have that 
\begin{eqnarray*}
\left|\Delta_e(\hyy) -\Phi_e(\hff)\right| &= & \left|\Delta_e(\yy)+\delta (\tphi_v-\tphi_u) -\Phi(\ff)- \left(\frac{1}{(\hu_e^+)^2}+\frac{1}{(\hu_e^-)^2}\right)\delta \tf_e - (\delta \tf_e)^2\zeta_e\right|\\
&\leq & \left|\delta \left((\tphi_v-\tphi_u) - r_e \tf_e \right) - (\delta \tf_e)^2\zeta_e\right| + \frac{\gamma_e}{\hu_e(\ff)}\\
&=& \left|  (\delta \tf_e)^2\zeta_e\right| + \frac{\gamma_e}{\hu_e(\ff)} \leq \frac{7 \delta^2 \frac{\tf_e^2}{\hu_e^2} + \frac{4}{3}\gamma_e}{\hu_e(\hff)} = \frac{7 (\delta \rho_e)^2 + \frac{4}{3}\gamma_e}{\hu_e(\hff)},
\end{eqnarray*}
where the second to last equality follows by Ohm's law \eqref{eq:def_energy_potentials}, the last equality follows by  \eqref{eq:congestion_vec_def} and we also used the fact that $\hu_e(\hff)\geq \frac{3}{4}\hu_e(\ff)$ since $\delta \leq (4\inorm{\vrho})^{-1}$. The lemma follows.

\subsection{Proof of Lemma \ref{lem:fixing_step}}\label{app:fixing_step}

For each arc $e$, let us define
\begin{equation}\label{eq:delta_def}
\theta_e:=\left(\frac{1}{(\hu_e^+(\gg))^2}+\frac{1}{(\hu_e^-(\gg))^2}\right)^{-1}\left(\Delta_e(\zz) - \Phi_e(\gg)\right).
\end{equation}
Intuitively, $\theta_e$ is the first-order correction to the coupling of the primal dual solution $(\gg, \zz)$ with respect to the arc $e$.

Note that by the fact that $(\gg,\zz)$ is $\vvarsigma$-coupled we have that
\begin{equation}\label{eq:bound_on_theta}
|\theta_e|\leq \left(\frac{1}{(\hu_e^+(\gg))^2}+\frac{1}{(\hu_e^-(\gg))^2}\right)^{-1}\left|\Delta_e(\zz) - \Phi_e(\gg)\right|\leq \hu_e(\gg)^2 \left(\frac{\varsigma_e}{\hu_e(\gg)}\right)=\varsigma_e \hu_e(\gg).
\end{equation}

 Let us now define a flow $\gg'$ to be the flow $\gg$ with the ``correcting`` flow $\vtheta$ added. That is, 
\[
g'_e:=g_e+\theta_e,
\]
for each arc $e$. Observe that, as $\inorm{\vvarsigma}\leq \norm{\vvarsigma}{2}\leq \frac{1}{50}$ and, by \eqref{eq:bound_on_theta}, $|\theta_e|\leq \varsigma_e \hu_e(\gg)$, this flow $\gg'$ is feasible in $G_{\gg}$ and, in fact, $\frac{51}{50}\hu_e(\gg)\geq \hu_e(\gg')\geq \frac{49}{50}\hu_e(\gg)$.

Furthermore, if we analyze the coupling of $(\gg',\zz)$, by Fact \ref{fa:Taylor_approximation}, for each arc $e$, we get that
\begin{eqnarray*}
\left| \Delta_e(\zz) - \Phi_e(\gg')\right|& = &\left| \Delta_e(\zz) - \left(\frac{1}{\hu_e^{+}(\gg)-\theta_e}-\frac{1}{\hu_e^{-}(\gg)+\theta_e}\right)\right|\\
 &= &\left| \Delta_e(\zz) - \Phi_e(\gg) - \left(\frac{1}{(\hu_e^+(\gg))^2}+\frac{1}{(\hu_e^-(\gg))^2}\right)\theta_e - \theta_e^2 \zeta_e'\right|\\
&=& \left|- \theta_e^2 \zeta_e'\right|\leq \frac{5\hu_e(\gg)^2 \varsigma_e^2}{\hu_e(\gg)^3}=\frac{5\varsigma_e^2}{\hu_e(\gg)}\leq \frac{51\varsigma_e^2}{10\hu_e(\gg')},
\end{eqnarray*}
where the third equality follows by the definition of $\theta_e$ \eqref{eq:delta_def} and the second to last inequality follows by \eqref{eq:bound_on_theta}.
 
We thus see that the solutions $(\gg', \zz)$ are well-coupled. In fact, they are $\tvgamma$-coupled with 
\begin{equation}\label{eq:bound_on_tgamma}
\norm{\tvgamma}{2}=\sqrt{\sum_e \tgamma_e^2}=\sqrt{\sum_e \left(\frac{11\varsigma_e^2}{2}\right)^2}\leq \frac{51}{10}\norm{\varsigma}{2}^2\leq \frac{51}{25000}.
\end{equation}

So, these solutions $(\gg', \zz)$ have all the desired properties except $\gg'$ is an $(\vchi_{\alpha'} + \hvsigma)$-flow, where $\hvsigma$ are the demands of the ``correcting`` flow $\vtheta$ we added -- and not a $\vchi_{\alpha'}$-flow we need.

To remedy this deficiency, we compute an electrical $(-\hvsigma)$-flow $\hvtheta$ determined by the to resistances 
\begin{equation}\label{eq:def_of_hr}
\hr_e:=\frac{1}{(\hu_e^+(\gg'))^2}+\frac{1}{(\hu_e^-(\gg'))^2},
\end{equation}
and obtain our desired solutions $(\ogg, \ozz)$ by adding this electrical flow to the flow $\gg'$, and simultaneously adding the vertex potentials $\hvphi$ that induce $\hvtheta$ to the dual solution $\zz$. In other words, we set
\begin{eqnarray}\label{eq:fixing_update}
\ogg & := & \gg' + \hvtheta\\
\ozz & := & \zz + \hvphi.\nonumber
\end{eqnarray}
Clearly, $\ogg$ is a $\vchi_{\alpha'}$-flow, as desired. It thus only remains to ensure that the solutions $(\ogg, \ozz)$ are still well-coupled.

To analyze the coupling of $(\ogg, \ozz)$ one should note first that the procedure we used to obtain these solutions from the solutions $(\gg', \zz)$ is analogous to the one we used obtain the solutions $(\hff, \hyy)$ from the solutions $(\ff, \yy)$ in our augmentation step. More precisely, the update \eqref{eq:fixing_update} is an exact analogue of the update \eqref{eq:augmentation_step_update} with $\delta=1$. 

So, to quantify the change in coupling we can follow the approach we used when analyzing the augmentation step. Specifically, let us define a congestion vector $\hvrho$ as
\[
\hrho_e := \frac{\htheta_e}{\hu_e(\gg')},
\]
for each arc $e$. By Lemma \ref{lem:energy_rho}, we know that
\[
\norm{\hvrho}{2}^2 \leq \energy{\hrr}{\hvtheta} \leq \energy{\hrr}{-\vtheta}, 
\]
where the second inequality follows by the energy minimizing property of electrical flows. (Note that $-\vtheta$ is a $(-\hvsigma)$-flow.) We thus have that, by \eqref{eq:def_of_hr},
\begin{eqnarray*}
\energy{\hrr}{-\vtheta} &=& \sum_e \hr_e (-\theta_e)^2 = \sum_e \left(\frac{1}{(\hu_e^+(\gg'))^2}+\frac{1}{(\hu_e^-(\gg'))^2}\right) \theta_e^2 \\
& \leq & \sum_e \frac{2}{\hu_e(\gg')^2} \left(\varsigma_e \hu_e(\gg)\right)^2 \leq \sum_e 2\left(\frac{51 \varsigma_e}{50}\right)^2\leq \frac{21}{20} \norm{\varsigma}{2}^2\leq \frac{1}{2000},
\end{eqnarray*}
where the first inequality follows by \eqref{eq:bound_on_theta}. As a result, we can conclude that
\begin{equation}\label{eq:bound_on_hrho}
\norm{\hvrho}{2}^2 \leq \energy{\hrr}{-\vtheta} \leq \frac{1}{2000},
\end{equation}
which ensures, in particular, that our analogue of condition \eqref{eq:l_infty_bound} (for $\delta=1$) is satisfied. 

Consequently, we can use Lemma \ref{lem:augmentation_coupling} to conclude that the solutions $(\ogg,\ozz)$ is $\hvgamma$-coupled with 
\[
\norm{\hvgamma}{2}=\sqrt{\sum_e \hgamma_e^2}\leq \sqrt{\sum_e \left(\frac{4}{3}\tgamma_e + 7 \hrho_e^2\right)^2}\leq \frac{4}{3} \norm{\tvgamma}{2}+ 7 \norm{\hvrho}{2}^2 \leq \frac{204}{75000} + \frac{7}{2000}<\frac{1}{100},
\] 
where we used the estimates \eqref{eq:bound_on_tgamma} and \eqref{eq:bound_on_hrho}. So, $(\ogg,\ozz)$ is well-coupled, as desired.

\subsection{Proof of Lemma \ref{lem:symmetrization_throughput}}\label{app:symmetrization_throughput}

Let us assume that $\vchi^T\yy\leq \frac{2m}{(1-\alpha)}$. Otherwise, by Lemma \ref{lem:duality_certificate}, we have a dual certificate that the target demand $\vchi$ cannot be routed in $G$, i.e., the condition (b) triggers.

As all the preconditioning arcs have $s$ and $t$ as their endpoints and are oriented from $s$ to $t$, we have that
\begin{equation}\label{eq:delta_preconditioning}
\Delta_e (\yy) = y_t-y_s = \frac{\vchi^T \yy}{F} \leq \frac{2m}{(1-\alpha)F},
\end{equation}
for each such arc $e$. Also, all these preconditioning arcs are completely indistinguishable from the point of view of our algorithm. They have the same endpoints, orientation, and initial (residual) capacities and stretch. Consequently, their resistances are always equal when we compute our electrical flows (see, e.g., \eqref{eq:res_definition}) and thus, in turn, their forward and backward residual capacities evolve in exactly the same way during each of the augmentation and fixing steps (see Section \ref{sec:progress_steps}). 

We want to argue that, for each such arc $e$,
\begin{equation}
\label{eq:precondition_arc_capacity_bound}
\hu_{e}(\ff)\geq \frac{(1-\alpha)F}{5m}.
\end{equation} 
As the residual capacity $\hu_e(\ff)$ is exactly the capacity of $e$ in the symmetrization $\hG_{\ff}$ of the residual graph $G_{\ff}$, and preconditioning arc constitute exactly half of all the $m$ arcs, a flow $\ff'$ that just routes $\frac{(1-\alpha)F}{5m}$ units of flow from $s$ to $t$ on each such arc will be the desired $\vchi_{\frac{(1-\alpha)}{10}}$-flow feasible in $\hG_{\ff}$, and our proof will follow.

To establish \eqref{eq:precondition_arc_capacity_bound}, let us fix some representative preconditioning arc $e$ and assume for the sake of contradiction that \eqref{eq:precondition_arc_capacity_bound} does not hold. Since the total preconditioning arc capacity constitutes at least $\frac{2}{3}$-fraction of the total $s$-$t$ capacity $F^*$ of the graph $G$, we must have that
\begin{equation}\label{eq:precond_u_plus_bound}
u^+_e=u_e^-\geq \frac{2}{m}\cdot \frac{2}{3} F^*  = \frac{4F^*}{3m},
\end{equation}
where $u^+_e$ and $u_e^-$ are the initial capacities of the preconditioning arcs. As a result, if \eqref{eq:precondition_arc_capacity_bound} did not hold, by \eqref{eq:capacities_residual_graph}, we know that
\[
|f_e|\geq \max\{u_e^+-\hu_e^+(\ff), u_e^--\hu_e^-(\ff)\}\geq \frac{4F^*}{3m}- \frac{(1-\alpha)F}{5m}\geq \frac{\frac{5}{2}F^*+\alpha F}{3m}> \frac{2F^*}{3m},
\]
where we used the fact that $F\leq \frac{3}{2}F^*$ as otherwise we would be able to immediately tell that it is impossible to route at least $\frac{2}{3}$-fraction of the $s$-$t$ flow only via preconditioning arcs. 

In fact, we also know that $f_e>0$. This is so, as $e$ is oriented from $s$ to $t$ and thus having $f_e<0$ would mean that there is a flow of strictly more than $\frac{F^*}{3}$ from $t$ towards $s$ over all the preconditioning arcs. It would, however, be impossible to balance this ``backward'' $t$-$s$ flow via an $s$-$t$ flow over the remaining arcs as they can support an $s$-$t$ flow of at most $\frac{F^*}{3}$. So, indeed $f_e>0$.

Furthermore, if \eqref{eq:precondition_arc_capacity_bound} would indeed not hold then $f_e>0$, \eqref{eq:capacities_residual_graph} and \eqref{eq:precond_u_plus_bound} would imply that
\[
\hu_e^+=\hu_e(\ff)< \frac{(1-\alpha)F}{3m} \leq \frac{F^*}{2m} \leq \frac{1}{2} u_e^- = \frac{1}{2} (\hu_e^-(\ff)-f_e) \leq \frac{\hu_e^-(\ff)}{2},
\]
where we again used the fact that $F\leq \frac{3}{2} F^*$. 

Now, combining the above bound with the fact that $(\ff, \yy)$ is well-coupled \eqref{eq:coupling}, we can conclude that, for any preconditioning arc $e$,
\begin{eqnarray*}
\Delta_e (\yy) &\geq &  \Phi_e(\ff) - \frac{\gamma_e}{\hu_e(\ff)} =\frac{1}{\hu_e^+(\ff)}-\frac{1}{\hu_e^-(\ff)}- \frac{\gamma_e}{\hu_e(\ff)} \geq \frac{(1-\gamma_e)}{\hu_e(\ff)}-\frac{1}{2\hu_e(\ff)}\\
&= & \frac{(1-2\gamma_e)}{2\hu_e(\ff)} \geq \frac{5m(1-\frac{2}{\sqrt{m}})}{2(1-\alpha)F}>\frac{2m}{(1-\alpha)F},
\end{eqnarray*}
where we again used the assumption that \eqref{eq:precondition_arc_capacity_bound} does not hold as well as the fact that $\gamma_e\leq \frac{1}{\sqrt{m}}$ since otherwise the contribution of all the $\frac{m}{2}$ preconditioning arcs  to the $\ell_2$-norm of the violation vector $\vgamma$ alone would make this norm violate the well-coupling bound. However, the above lower bound on $\Delta_e(\yy)$ directly violates our bound \eqref{eq:delta_preconditioning}. So, we reached a contradiction that proves \eqref{eq:precondition_arc_capacity_bound} and thus our lemma follows.

\subsection{Proof of Lemma \ref{lem:energy_decrease_control}}\label{app:energy_decrease_control}

By examining the augmentation and fixing steps that are performed during the progress steps (cf. \eqref{eq:augmentation_step_update} and \eqref{eq:fixing_update} in the proof of Lemma \ref{lem:fixing_step}) and applying standard Taylor approximations to the resulting changes in the resistances $\rr$ (cf. \eqref{eq:res_definition}), one obtains that, for each arc $e$, the resulting arc resistance $r_e'$ change is such that
\begin{equation}\label{eq:change_r_e_progress}
r_e' \leq \left(1+ O(\delta |\rho_e| + \kappa_e)\right)^{-1} r_e,
\end{equation}
where $\vrho$ is the congestion vector \eqref{eq:congestion_vec_def}, $\delta=(33\norm{\vrho}{3})^{-1}$ is the step size, and $\vkappa$ is a vector such that $\norm{\vkappa}{2}\leq 1$. (Roughly speaking, the $\delta \rho_e$ term corresponds to augmentation step and the $\kappa_e$ term corresponds to the fixing step.)

Now, let $\tphi$ be the vertex potentials inducing the electrical $\vchi$-flow $\tff$. By Lemma \ref{lem:effective_conductance}, we know that
\[
\frac{1}{\energy{\rr}{\tff}} = \sum_{e=(u,v)} \frac{(\tphi_v-\tphi_u)^2}{r_e \energy{\rr}{\tff}^2}.
\]
Furthermore, if $\energy{\rr'}{\tff'}$ is the energy of an electrical $\vchi$-flow $\tff'$ determined by the new resistances $\rr'$, applying Lemma \ref{lem:effective_conductance} again, we obtain that
\begin{eqnarray*}
\frac{1}{\energy{\rr'}{\tff'}} &\leq & \sum_{e=(u,v)} \frac{(\tphi_v-\tphi_u)^2}{r_e' \energy{\rr}{\tff}^2}\leq \sum_{e=(u,v)} \frac{\left(1+ O(\delta |\rho_e| + \kappa_e)\right)(\tphi_v-\tphi_u)^2}{r_e \energy{\rr}{\tff}^2}\\
&=& \frac{1}{\energy{\rr}{\tff}} \left(1+O\left(\sum_e \frac{(\delta \rho_e+\kappa_e)(r_e \tf_e^2)}{\energy{\rr}{\tff}}\right)\right)\\
& \leq & \frac{1}{\energy{\rr}{\tff}} \left(1+O\left(\sum_e \frac{(\delta \rho_e^3+\kappa_e \rho_e^2)}{\norm{\vrho}{2}^2}\right)\right),
\end{eqnarray*}
where we used \eqref{eq:change_r_e_progress}, Ohm's law \eqref{eq:Ohms_lawy} and Lemmas \ref{lem:energy_rho} and \ref{lem:energy_upperbound}.

Finally, observe that since $\delta=(33\norm{\vrho}{3})^{-1}$ and $\norm{\vkappa}{2}\leq 1$,  an application of the Cauchy-Schwartz inequality gives us that
\[
\sum_e (\delta \rho_e^3+\kappa_e \rho_e^2) \leq \delta \norm{\vrho}{3}^3 + \norm{\vkappa}{2}\norm{\vrho}{4}^2 \leq \delta \norm{\vrho}{3}^3 + \norm{\vkappa}{2}\norm{\vrho}{3}^2 \leq O(\norm{\vrho}{3}^2),
\]
as desired. The lemma thus follows.

\bibliographystyle{abbrv}
\bibliography{../main}

\end{document}